\documentclass[pdflatex,sn-mathphys-num]{sn-jnl}

\usepackage{graphicx}
\usepackage{float}
\usepackage{multirow}
\usepackage{amsmath,amssymb,amsfonts}
\usepackage{amsthm}
\usepackage{mathrsfs}
\usepackage{mathtools}
\usepackage[title]{appendix}
\usepackage{xcolor}
\usepackage{subcaption}
\usepackage{textcomp}
\usepackage{manyfoot}
\usepackage{booktabs}
\usepackage{algorithm}
\usepackage{algorithmicx}
\usepackage{algpseudocode}
\usepackage{listings}

\lstdefinestyle{bashstyle}{
    language=bash,
    backgroundcolor=\color{gray!8},
    basicstyle=\ttfamily\footnotesize,
    keywordstyle=\color{blue!60!black}\bfseries,
    commentstyle=\color{gray!50!black}\itshape,
    stringstyle=\color{teal!70!black},
    frame=single,
    rulecolor=\color{black!30},
    breaklines=true,
    breakatwhitespace=false,
    showstringspaces=false,
    tabsize=2,
    xleftmargin=1em,
}

\definecolor{pythonBackground}{RGB}{248,249,251}
\definecolor{pythonFrame}{RGB}{205,210,218}
\definecolor{pythonKeyword}{RGB}{33,85,138}
\definecolor{pythonComment}{RGB}{86,116,86}
\definecolor{pythonString}{RGB}{153,76,0}
\lstdefinestyle{pythonstyle}{
    language=Python,
    backgroundcolor=\color{pythonBackground},
    basicstyle=\ttfamily\scriptsize,
    keywordstyle=\color{pythonKeyword}\bfseries,
    commentstyle=\color{pythonComment}\itshape,
    stringstyle=\color{pythonString},
    numbers=left,
    numberstyle=\tiny\color{gray!70!black},
    numbersep=8pt,
    frame=single,
    rulecolor=\color{pythonFrame},
    framerule=0.4pt,
    framesep=5pt,
    breaklines=true,
    breakatwhitespace=false,
    showstringspaces=false,
    keepspaces=true,
    columns=fullflexible,
    tabsize=4,
    upquote=true,
    xleftmargin=2.4em,
    framexleftmargin=2em,
    captionpos=t,
    aboveskip=1em,
    belowskip=1em,
}

\numberwithin{equation}{section}

\usepackage{geometry}
\usepackage[T1]{fontenc}
\usepackage{lmodern}

\usepackage{tikz}
\usepackage{pgfplots}
\pgfplotsset{compat=1.18}

\usetikzlibrary{patterns}
\usetikzlibrary{arrows,arrows.meta}

\usepackage[nameinlink,noabbrev]{cleveref}
\hypersetup{hidelinks}
\theoremstyle{thmstyleone}
\newtheorem{theorem}{Theorem}

\theoremstyle{thmstyletwo}
\newtheorem{example}{Example}

\theoremstyle{thmstylethree}

\crefname{theorem}{theorem}{theorems}
\Crefname{theorem}{Theorem}{Theorems}

\crefname{proposition}{proposition}{propositions}
\Crefname{proposition}{Proposition}{Propositions}

\crefname{lemma}{lemma}{lemmas}
\Crefname{lemma}{Lemma}{Lemmas}

\crefname{corollary}{corollary}{corollaries}
\Crefname{corollary}{Corollary}{Corollaries}

\crefname{definition}{definition}{definitions}
\Crefname{definition}{Definition}{Definitions}

\crefname{example}{example}{examples}
\Crefname{example}{Example}{Examples}

\crefname{remark}{remark}{remarks}
\Crefname{remark}{Remark}{Remarks}

\newcommand{\mentry}[1]{\mathmakebox[1.25em][c]{#1}}
\definecolor{rhoLegendMarker}{HTML}{708087}
\newcommand{\rhoLegendMark}[2][]{%
	\tikz[baseline=-0.55ex]{%
		\draw plot[
			only marks,
			mark=#2,
			mark size=2.1pt,
			mark options={fill=rhoLegendMarker,draw=white,line width=0.25pt,#1}
		] coordinates {(0,0)};%
	}%
}
\newcommand{\rhoLegendPlus}{%
	\tikz[baseline=-0.55ex,x=1pt,y=1pt]{%
		\fill[rhoLegendMarker] (-2.5,-0.8) rectangle (2.5,0.8);
		\fill[rhoLegendMarker] (-0.8,-2.5) rectangle (0.8,2.5);
	}%
}
\newcommand{\sharedRhoLegend}{%
	\begingroup\footnotesize
		$\rho$\hspace{1em}%
		\rhoLegendMark{*}\,$2^{-1}$\hspace{1.3em}%
		\rhoLegendMark{square*}\,$2^{-2}$\hspace{1.3em}%
		\rhoLegendMark{diamond*}\,$2^{-3}$\hspace{1.3em}%
		\rhoLegendMark{triangle*}\,$2^{-4}$\hspace{1.3em}%
		\rhoLegendMark[rotate=180]{triangle*}\,$2^{-5}$\hspace{1.3em}%
		\rhoLegendPlus\,$2^{-6}$%
	\endgroup
}

\begin{document}

\title[A 52-Addition, Rank-23 Scheme for Exact $3 \times 3$ Matrix Multiplication]{A 52-Addition, Rank-23 Scheme for Exact $3 \times 3$ Matrix Multiplication}

\author*[1]{\fnm{Hugo Møller} \sur{Nielsen}}\email{hugomn2002@gmail.com}

\affil[1]{\orgdiv{Department of Applied Mathematics and Computer Science}, \orgname{Technical University of Denmark (DTU)},
	\orgaddress{\city{Kgs. Lyngby}, \postcode{2800}, \country{Denmark}}}

\abstract{
	We present a rank-$23$ algorithm for general $3 \times 3$ matrix multiplication using $52$
	additions or subtractions, together with a Karstadt--Schwartz-style basis change
	\cite{SchwartzVakninPebble,KarstadtSchwartz2020MatMulLittleFaster}.
	This gives, to the best of our knowledge, a new best-known candidate for the
	addition-minimization problem for rank-$23$ general $3 \times 3$ matrix multiplication.
}


\maketitle

\section{Introduction}
In $1969$, Strassen discovered the first sub-cubic matrix multiplication algorithm,
and variants of Strassen's algorithm are still used in high-performance matrix multiplication
implementations \cite{Strassen1969GaussianElimination}.
Since then, several works have reduced the constant factors in the running time of
Strassen's algorithm.
A recent improvement to the leading constant factor of Strassen's algorithm was obtained by
Karstadt and Schwartz in \cite{KarstadtSchwartz2020MatMulLittleFaster}. They introduced a
basis-change method that reduced the number of additions required at each
recursive step in Strassen's algorithm from $15$ to $12$.
Moreover, they proved that, within their framework, $12$ additions is optimal.

A natural next step beyond Strassen's rank-$7$ algorithm for $2 \times 2$ matrices is to study
Strassen-like recursive algorithms for $3 \times 3$ matrices.
In $1976$, Laderman introduced a rank-$23$ algorithm for multiplying $3 \times 3$
matrices \cite{Laderman1976Noncommutative3x3}. Despite extensive subsequent work,
no construction of lower rank is known for general $3 \times 3$ matrix multiplication
\cite{HeuleKauersSeidl2021NewWays3x3,AlphaTensor2022}.
As such, recent work has focused on reducing the addition count, and hence the leading
constant, for rank-$23$ matrix multiplication algorithms, see \Cref{fig:3x3_addition_progression}.

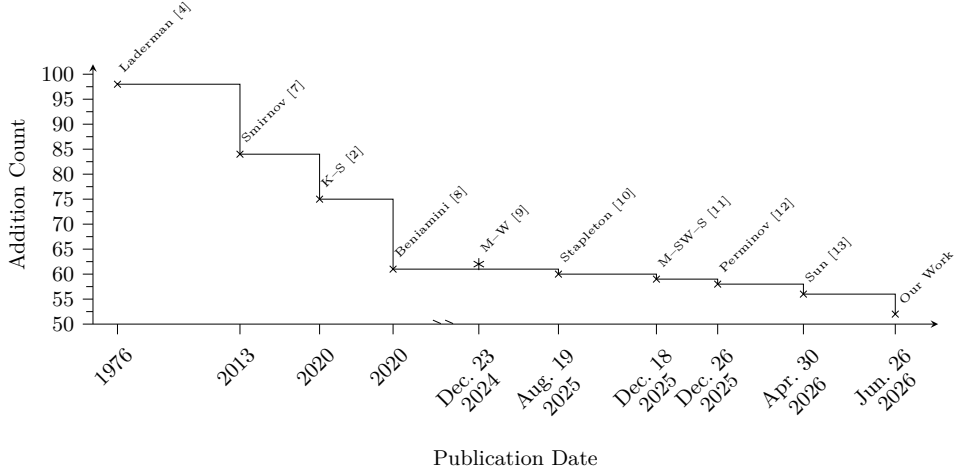
\begin{figure}[H]
	\begin{center}
		\begin{tikzpicture}[scale=0.9]
			\begin{axis}[
					width=14cm,
					height=5.4cm,
					xlabel={Publication Date},
					ylabel={Addition Count},
					axis lines=left,
					axis line style={black, thin},
					tick style={black, thin},
					tick align=outside,
					tick label style={font=\small},
					label style={font=\small},
					xmin=-0.4,
					xmax=13.4,
					ymin=50,
					ymax=102,
					ytick={50,55,60,65,70,75,80,85,90,95,100},
					minor y tick num=1,
					grid=none,
					clip=false,
					xtick={0,2,3.3,4.5,5.9,7.2,8.8,9.8,11.2,12.7},
					xticklabels={
							1976,
							2013,
							2020,
							2020,
							{\shortstack{Dec.~23\\2024}},
							{\shortstack{Aug.~19\\2025}},
							{\shortstack{Dec.~18\\2025}},
							{\shortstack{Dec.~26\\2025}},
							{\shortstack{Apr.~30\\2026}},
							{\shortstack{Jun.~26\\2026}}
						},
					xticklabel style={
							font=\scriptsize,
							rotate=45,
							anchor=north east,
							align=center,
							xshift=2pt,
							yshift=-2pt
						},
					xlabel style={yshift=-0.1em},
				]

				\tikzset{
					pointlabel/.style={
							font=\tiny,
							rotate=45,
							anchor=south west,
							inner sep=2pt,
							xshift=3pt,
							yshift=-2pt
						}
				}

				\addplot[
					black,
					thin,
					const plot,
					mark=none,
				] coordinates {
						(0,    98)
						(2,    84)
						(3.3,  75)
						(4.5,  61)
						(7.2,  60)
						(8.8,  59)
						(9.8,  58)
						(11.2, 56)
						(12.7, 52)
					};

				\addplot[
					only marks,
					mark=x,
					mark size=2pt,
					black,
					thin,
				] coordinates {
						(0,    98)
						(2,    84)
						(3.3,  75)
						(4.5,  61)
						(7.2,  60)
						(8.8,  59)
						(9.8,  58)
						(11.2, 56)
						(12.7, 52)
					};

				\addplot[
					only marks,
					mark=asterisk,
					mark size=2.5pt,
					black,
					thin,
				] coordinates {
						(5.9, 62)
					};

				\node[pointlabel] at (axis cs:0,98)      {Laderman \protect\cite{Laderman1976Noncommutative3x3}};
				\node[pointlabel] at (axis cs:2,84)      {Smirnov \protect\cite{Smirnov2013BilinearComplexity}};
				\node[pointlabel] at (axis cs:3.3,75)    {K--S \protect\cite{KarstadtSchwartz2020MatMulLittleFaster}};
				\node[pointlabel] at (axis cs:4.5,61)    {Beniamini \protect\cite{Beniamini2020SparsifyingOperators}};
				\node[pointlabel] at (axis cs:5.9,62)    {M--W \protect\cite{MartenssonWagner_NumberBeast}};
				\node[pointlabel] at (axis cs:7.2,60)    {Stapleton \protect\cite{Stapleton2025SixtyAdditionRank23}};
				\node[pointlabel] at (axis cs:8.8,59)    {M--SW--S \protect\cite{MartenssonWagnerStapleton2025Rank23Arithmetic59}};
				\node[pointlabel] at (axis cs:9.8,58)    {Perminov \protect\cite{Perminov2025FiftyEightAdditionRank23}};
				\node[pointlabel] at (axis cs:11.2,56)   {Sun \protect\cite{sun2026exact56additionrank23}};
				\node[pointlabel] at (axis cs:12.7,52)   {Our Work};

				\draw[thin]
				(axis cs:5.15,50.8) -- (axis cs:5.28,50.1)
				(axis cs:5.35,50.8) -- (axis cs:5.48,50.1);

			\end{axis}
		\end{tikzpicture}
	\end{center}
	\caption{Progress on reducing the number of additions for each recursive step of
		a rank-$23$ Strassen-type algorithm for general $3 \times 3$ matrix multiplication.}
	\label{fig:3x3_addition_progression}
\end{figure}
In this work, we adapt and improve the greedy-potential heuristic proposed by
Mårtensson and Wagner \cite{MartenssonWagner_NumberBeast}.
We first work over $\mathbb{F}_2$ and apply the heuristic to
the rank-$23$ dataset for general $3 \times 3$ matrix multiplication generated
by Heule, Kauers, and Seidel \cite{HeuleKauersSeidl2021NewWays3x3}.
To enlarge the search space, we apply Bernoulli-random basis-change matrices to each data point.
Finally, we lift our solutions to general fields $\mathbb{F}$ by replacing certain
additions over $\mathbb{F}_2$ with subtractions over $\mathbb{F}$.

\section{Background}
Throughout the paper, an \emph{addition} will refer to an operation of the form
\begin{align*}
	f : R^2 \to R, \qquad f(x,y) = \lambda x + \mu y,
\end{align*}
where $R$ is a ring and $\lambda, \mu \in \{-1, 1\}$. This is the standard
convention when counting additions in Strassen-type matrix multiplication
schemes \cite{MartenssonWagner_NumberBeast,Stapleton2025SixtyAdditionRank23}.

\subsection{Strassen-type schemes for square matrices}
Let $U,V,W \in \{-1,0,1\}^{r \times k^2}$. We say that $(U,V,W)$ is a rank-$r$ multiplication
scheme for $k \times k$ matrices if, for all $A,B \in R^{k \times k}$ over an arbitrary
ring $R$,
\begin{align} \label{eq:main_vectorized_strassen_type_eq}
	\vec{AB} = W^\intercal \left(\left(U\vec{A}\right) \odot \left(V\vec{B}\right)\right),
\end{align}
where $\odot$ is the Hadamard product and $\vec{X}$ denotes the column-wise
vectorization of the matrix $X \in R^{k \times k}$.
The rank $r$ is the number of multiplications over $R$ used by the scheme.
Applying such a scheme recursively to $k \times k$ block matrices yields a
Strassen-type algorithm with arithmetic complexity $O(n^{\log_{k}{r}})$,
see \Cref{thm:time_complexity_generalized_strassen}.

\begin{theorem} \label{thm:time_complexity_generalized_strassen}
	Let $k,n \in \mathbb{N}$ and $r > k^2$. Given a scheme of the form \cref{eq:main_vectorized_strassen_type_eq}
	with $W, U, V \in \{ -1, 0, 1 \}^{r \times k^2}$ that requires $\alpha$ additions to evaluate, we can multiply matrices
	$A,B \in \mathbb{F}^{n \times n}$ using $n^{\log_k r}$ multiplications and $\frac{\alpha}{r - k^2}\left(n^{\log_k{r}} - n^2\right)$
	additions when $n$ is a power of $k$. We use $O(n^{\log_{k}{r}})$ arithmetic operations for multiplying two $n \times n$ matrices
	for general $n \in \mathbb{N}$.
\end{theorem}
\begin{proof}
	Since $\mathbb{F}^{m \times m}$ is a ring for all $m \in \mathbb{N}$, our scheme applies to
	$k \times k$ block matrices with elements in $\mathbb{F}^{m \times m}$.

	Let $n = k^\ell$. We proceed by induction on $\ell$. For $\ell = 1$, the given $\alpha$-addition,
	$r$-multiplication scheme of the form in \cref{eq:main_vectorized_strassen_type_eq} establishes the base case.
	Suppose now that we require $n^{\log_k r} = r^\ell$ multiplications, and
	$\frac{\alpha}{r - k^2}\left(n^{\log_k{r}} - n^2\right) = \frac{\alpha}{r - k^2}\left(r^\ell - k^{2\ell}\right)$
	additions to compute the matrix product for some $\ell \in \mathbb{N}$. Given any two matrices
	$A,B \in \mathbb{F}^{k^{\ell + 1} \times k^{\ell + 1}}$, we can view them as $k \times k$ block matrices
	with entries in $\mathbb{F}^{k^\ell \times k^\ell}$.
	Applying \cref{eq:main_vectorized_strassen_type_eq} at the block level uses $r$ block multiplications,
	and $\alpha$ block additions. Each block multiplication uses $r^\ell$ multiplications and
	$\frac{\alpha}{r - k^2}\left(r^\ell - k^{2\ell}\right)$ additions, by the induction hypothesis,
	hence the block multiplications contribute a total of $r \cdot r^{\ell} = r^{\ell + 1}$ multiplications, and
	$\frac{r\alpha}{r - k^2}\left(r^\ell - k^{2\ell}\right)$ additions. Each block addition contributes
	$(k^\ell)^2 = k^{2\ell}$ additions. The total number of multiplications required to compute
	the matrix product $AB$ for matrices $A,B \in \mathbb{F}^{k^{\ell + 1} \times k^{\ell + 1}}$ is thus:
	\begin{align}
		r^{\ell + 1} = (kn)^{\log_k{r}},
	\end{align}
	and the total number of additions required is
	\begin{align}
		\frac{r\alpha}{r - k^2}\left(r^\ell - k^{2\ell}\right) + \alpha k^{2\ell}
		 & = \frac{r \alpha}{r - k^2}\left(r^\ell - k^{2\ell} + \frac{r - k^2}{r}k^{2\ell}\right) \\
		 & = \frac{\alpha}{r - k^2}\left(r^{\ell+1} - k^{2(\ell+1)}\right)
		= \frac{\alpha}{r - k^2} \left((kn)^{\log_k{r}} - (kn)^2\right),
	\end{align}
	which by the principle of mathematical induction completes the proof.

	Suppose now that $n$ is not a power of $k$. Pad $A,B$ with zeros to the size
	$m = k^{\lceil \log_k{n} \rceil}$.
	Denote the padded matrices $\tilde{A}$ and $\tilde{B}$ respectively, so that:
	\begin{align*}
		 & \tilde{A} = \begin{bmatrix}
			               A                    & 0_{n \times (m - n)}     \\
			               0_{(m - n) \times n} & 0_{(m-n) \times (m - n)}
		               \end{bmatrix},
		 & \tilde{B} = \begin{bmatrix}
			               B                    & 0_{n \times (m - n)}     \\
			               0_{(m - n) \times n} & 0_{(m-n) \times (m - n)}
		               \end{bmatrix},
	\end{align*}
	where $0_{a \times b}$ represents the $a \times b$ zero matrix.
	Since $m < kn$ and $m$ is a power of $k$, we can compute $\tilde{A}\tilde{B}$ in
	$O\left(m^{\log_k{r}}\right) = O\left((nk)^{\log_k{r}}\right) = O\left(n^{\log_k{r}}\right)$
	arithmetic operations. Finally, we can extract $AB$ from $\tilde{A}\tilde{B}$ as the top-left
	$n \times n$ submatrix of $\tilde{A}\tilde{B}$. We thus can compute the matrix product $AB$ using
	a total of $O(n^{\log_{k}{r}})$ arithmetic operations.
\end{proof}

For fixed coefficient matrices $U,V,W$ and $M = (U\overrightarrow{A}) \odot (V\overrightarrow{B})$, minimizing
the total number of additions required to evaluate each of the three linear maps
\begin{align*}
	U\overrightarrow{A}, \qquad V\overrightarrow{B}, \qquad W^{\top}M,
\end{align*}
directly reduces the leading constant in the asymptotic arithmetic operation count of the
corresponding Strassen-type recursive algorithm, see \Cref{thm:time_complexity_generalized_strassen}.
One way to reduce the number of additions required to compute each of these linear maps
is to search for new matrix triples $U,V,W \in \{-1, 0, 1\}^{r \times k^2}$ that require
fewer additions to compute. Another method is to efficiently reuse additions that have
already been computed, see \Cref{example:UVW_min}.

\begin{example} \label{example:UVW_min}
	Let $U,V,W \in \{-1,0,1\}^{7 \times 4}$ be the Winograd
	coefficient matrices \cite{fischer1974further,boyer2009memory} given by
	\begin{align}
		U & =
		\begin{bmatrix}
			\mentry{1}  & \mentry{0}  & \mentry{0} & \mentry{0}  \\
			\mentry{-1} & \mentry{1}  & \mentry{0} & \mentry{0}  \\
			\mentry{0}  & \mentry{1}  & \mentry{0} & \mentry{1}  \\
			\mentry{-1} & \mentry{1}  & \mentry{0} & \mentry{1}  \\
			\mentry{0}  & \mentry{0}  & \mentry{1} & \mentry{0}  \\
			\mentry{1}  & \mentry{-1} & \mentry{1} & \mentry{-1} \\
			\mentry{0}  & \mentry{0}  & \mentry{0} & \mentry{1}
		\end{bmatrix},
		\qquad
		V =
		\begin{bmatrix}
			\mentry{1}  & \mentry{0} & \mentry{0}  & \mentry{0}  \\
			\mentry{0}  & \mentry{0} & \mentry{1}  & \mentry{-1} \\
			\mentry{-1} & \mentry{0} & \mentry{1}  & \mentry{0}  \\
			\mentry{1}  & \mentry{0} & \mentry{-1} & \mentry{1}  \\
			\mentry{0}  & \mentry{1} & \mentry{0}  & \mentry{0}  \\
			\mentry{0}  & \mentry{0} & \mentry{0}  & \mentry{1}  \\
			\mentry{-1} & \mentry{1} & \mentry{1}  & \mentry{-1}
		\end{bmatrix},
		\qquad
		W =
		\begin{bmatrix}
			\mentry{1} & \mentry{1} & \mentry{1} & \mentry{1} \\
			\mentry{0} & \mentry{1} & \mentry{0} & \mentry{1} \\
			\mentry{0} & \mentry{0} & \mentry{1} & \mentry{1} \\
			\mentry{0} & \mentry{1} & \mentry{1} & \mentry{1} \\
			\mentry{1} & \mentry{0} & \mentry{0} & \mentry{0} \\
			\mentry{0} & \mentry{0} & \mentry{1} & \mentry{0} \\
			\mentry{0} & \mentry{1} & \mentry{0} & \mentry{0}
		\end{bmatrix}.
	\end{align}
	Expanding the right-hand side of \cref{eq:main_vectorized_strassen_type_eq}
	with this choice of $U$, $V$, and $W$ shows that it computes the matrix product
	$AB$ for any $A,B \in R^{2 \times 2}$.

	The linear transformation induced by $U$ on
	$\vec{A} = \begin{bmatrix} A_{11} & A_{21} & A_{12} & A_{22} \end{bmatrix}^{\intercal}$
	and the linear transformation induced by $V$ on
	$\vec{B} = \begin{bmatrix} B_{11} & B_{21} & B_{12} & B_{22} \end{bmatrix}^{\intercal}$
	can each be computed using $4$ additions as follows:
	\begin{align*}
		\begin{alignedat}{4}
			u_{1} & = A_{11}, & \qquad u_{2} & = A_{21} - A_{11}, & \qquad u_{3} & = A_{21} + A_{22}, & \qquad u_{4} & = u_{2} + A_{22}, \\
			u_{5} & = A_{12}, & \qquad u_{6} & = A_{12} - u_{4},  & \qquad u_{7} & = A_{22},          &              &                   \\
			v_{1} & = B_{11}, & \qquad v_{2} & = B_{12} - B_{22}, & \qquad v_{3} & = B_{12} - B_{11}, & \qquad v_{4} & = B_{11} - v_{2}, \\
			v_{5} & = B_{21}, & \qquad v_{6} & = B_{22},          & \qquad v_{7} & = B_{21} - v_{4},  &              &
		\end{alignedat}
	\end{align*}
	where $u_i = U_{i,:}\vec{A}$, and $v_i = V_{i,:}\vec{B}$ for each $i = 1,\dots,7$.

	Now define $M_i = u_i v_i$, for $i = 1,\dots,7$. The output transformation induced
	by $W^\intercal$ on $M \in R^{7}$ can be computed in $7$ additions by introducing
	the temporary variables $t_{1}$, $t_{2}$, and $t_{3}$ as follows:
	\begin{align*}
		 & t_{1} = M_{1} + M_{4},
		 &
		 & w_{1} = M_{1} + M_{5}, \\
		 & t_{2} = t_{1} + M_{2},
		 &
		 & w_{2} = t_{2} + M_{7}, \\
		 & t_{3} = t_{1} + M_{3},
		 &
		 & w_{3} = t_{3} + M_{6}, \\
		 & w_{4} = t_{2} + M_{3}.
	\end{align*}
	In total, this Winograd scheme computes $AB$ using $7$ multiplications and
	$4 + 4 + 7 = 15$ additions. \Cref{thm:time_complexity_generalized_strassen}
	thus implies a leading constant in the asymptotic operation count of
	$1 + \frac{15}{7 - 2^2} = 6$ for Winograd's scheme, improving upon Strassen's scheme
	which requires $18$ additions and hence has a leading constant of $1 + \frac{18}{7 - 2^2} = 7$.

	One way to represent the computations of the linear transformations induced by
	$U$, $V$, and $W$ is as edge-weighted directed acyclic graphs (DAGs),
	see \Cref{fig:DAG_rep_minadd_example}.
	\begin{figure}[H]
		\centering
		\resizebox{\textwidth}{!}{%
			\begin{tikzpicture}[
				vertex/.style={
						shape=circle,
						draw,
						minimum size=8mm,
						text width=8mm,
						align=center,
						inner sep=0pt,
						font=\scriptsize
					},
				edge/.style={-{Latex[length=2mm]}},
				lab/.style={pos=0.58, fill=white, inner sep=1pt, font=\scriptsize}
				]

				\begin{scope}[shift={(0,0)}]
					\node at (0,0.9) {$U\vec{A}$};

					\node[vertex] (Ua11) at (-2.7,0) {$A_{11}$};
					\node[vertex] (Ua21) at (-0.9,0) {$A_{21}$};
					\node[vertex] (Ua12) at ( 2.7,0) {$A_{12}$};
					\node[vertex] (Ua22) at ( 0.9,0) {$A_{22}$};

					\node[vertex] (Uu2) at (-1.8,-1.4) {$u_2$};
					\node[vertex] (Uu3) at (-0.2,-1.4) {$u_3$};

					\node[vertex] (Uu4) at (0.9,-2.8) {$u_4$};
					\node[vertex] (Uu6) at ( 2.7,-1.4) {$u_6$};

					\draw[edge] (Ua21) -- node[lab] {$+$} (Uu2);
					\draw[edge] (Ua11) -- node[lab] {$-$} (Uu2);

					\draw[edge] (Ua21) -- node[lab] {$+$} (Uu3);
					\draw[edge] (Ua22) -- node[lab] {$+$} (Uu3);

					\draw[edge] (Uu2) -- node[lab] {$+$} (Uu4);
					\draw[edge] (Ua22) -- node[lab] {$+$} (Uu4);

					\draw[edge] (Ua12) -- node[lab] {$+$} (Uu6);
					\draw[edge] (Uu4) -- node[lab] {$-$} (Uu6);
				\end{scope}

				\begin{scope}[shift={(7.8,0)}]
					\node at (0,0.9) {$V\vec{B}$};

					\node[vertex] (Vb11) at (-2.7,0) {$B_{11}$};
					\node[vertex] (Vb21) at ( 2.7,0) {$B_{21}$};
					\node[vertex] (Vb12) at (-0.9,0) {$B_{12}$};
					\node[vertex] (Vb22) at ( 0.9,0) {$B_{22}$};

					\node[vertex] (Vv3) at (-1.4,-1.4) {$v_3$};
					\node[vertex] (Vv2) at ( 0.4,-1.4) {$v_2$};

					\node[vertex] (Vv4) at (-2.7,-2.8) {$v_4$};
					\node[vertex] (Vv7) at ( 0.0,-4.2) {$v_7$};

					\draw[edge] (Vb12) -- node[lab] {$+$} (Vv2);
					\draw[edge] (Vb22) -- node[lab] {$-$} (Vv2);

					\draw[edge] (Vb12) -- node[lab] {$+$} (Vv3);
					\draw[edge] (Vb11) -- node[lab] {$-$} (Vv3);

					\draw[edge] (Vb11) -- node[lab] {$+$} (Vv4);
					\draw[edge] (Vv2) -- node[lab] {$-$} (Vv4);

					\draw[edge] (Vb21) -- node[lab] {$+$} (Vv7);
					\draw[edge] (Vv4) -- node[lab] {$-$} (Vv7);
				\end{scope}

				\begin{scope}[shift={(15.8,0)}]
					\node at (0,0.9) {$W^\intercal M$};
					\tikzset{
						halo/.style={
								preaction={draw=white, line width=4pt}
							}
					}

					\node[vertex] (Wm7) at (-4.2,0) {$M_7$};
					\node[vertex] (Wm2) at (-2.8,0) {$M_2$};
					\node[vertex] (Wm4) at (-1.4,0) {$M_4$};
					\node[vertex] (Wm1) at ( 0.0,0) {$M_1$};
					\node[vertex] (Wm5) at ( 1.4,0) {$M_5$};
					\node[vertex] (Wm3) at ( 2.8,0) {$M_3$};
					\node[vertex] (Wm6) at ( 4.2,0) {$M_6$};

					\node[vertex] (Wt1) at (-0.7,-1.4) {$t_1$};
					\node[vertex] (Ww1) at ( 0.7,-1.4) {$w_1$};

					\node[vertex] (Wt2) at (-2.1,-2.8) {$t_2$};
					\node[vertex] (Wt3) at ( 1.4,-2.8) {$t_3$};

					\node[vertex] (Ww2) at (-3.15,-4.2) {$w_2$};
					\node[vertex] (Ww4) at (-0.75,-4.2) {$w_4$};
					\node[vertex] (Ww3) at ( 2.8,-4.2) {$w_3$};

					\draw[edge] (Wm4) -- node[lab] {$+$} (Wt1);
					\draw[edge] (Wm1) -- node[lab] {$+$} (Wt1);

					\draw[edge] (Wm1) -- node[lab] {$+$} (Ww1);
					\draw[edge] (Wm5) -- node[lab] {$+$} (Ww1);

					\draw[edge] (Wt1) -- node[lab] {$+$} (Wt2);
					\draw[edge] (Wm2) -- node[lab] {$+$} (Wt2);

					\draw[edge] (Wt1) -- node[lab] {$+$} (Wt3);
					\draw[edge] (Wm3) -- node[lab] {$+$} (Wt3);

					\draw[edge] (Wt2) -- node[lab] {$+$} (Ww2);
					\draw[edge] (Wm7) -- node[lab] {$+$} (Ww2);

					\draw[edge] (Wt3) -- node[lab] {$+$} (Ww3);
					\draw[edge] (Wm6) -- node[lab] {$+$} (Ww3);

					\draw[edge] (Wt2) -- node[lab] {$+$} (Ww4);

					\draw[edge, halo]
					(Wm3.south) .. controls +(0.25,-2.45) and +(2.45,0.35) ..
					node[lab, pos=0.72] {$+$} (Ww4.east);
				\end{scope}

			\end{tikzpicture}%
		}
		\caption{DAG representations of the minimal-addition computations of $U\vec A$, $V\vec B$, and $W^\intercal M$
			from \protect\Cref{example:UVW_min}.
			Every non-input node represents the signed sum induced by its incoming edges.
			Edges marked with "$+$" indicate addition of the source node, while edges marked
			with "$-$" indicate addition of the negated source node.
			Only nontrivial computations are shown, i.e. direct copies such
			as $u_{1} = A_{11}$ are omitted.}
		\label{fig:DAG_rep_minadd_example}
	\end{figure}
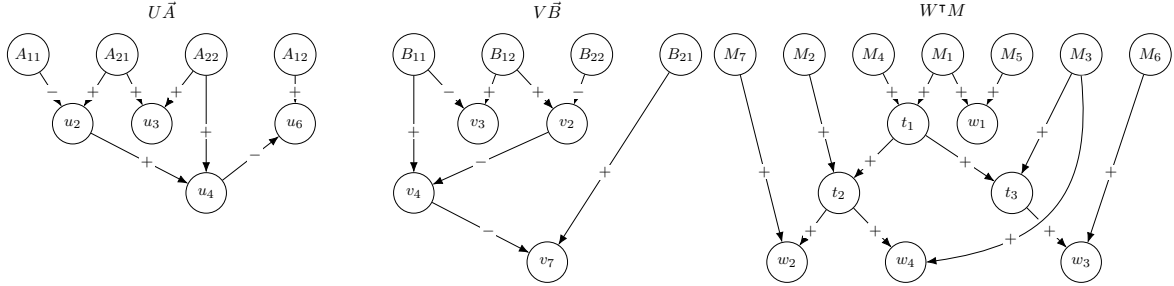
\end{example}
In this work, we will not explicitly find new triples of matrices $(U,V,W)$ satisfying
\cref{eq:main_vectorized_strassen_type_eq}. Our main focus is the efficient reuse of
additions, see \Cref{example:UVW_min}.

\subsection{Basis Change}
The condition implied by \cref{eq:main_vectorized_strassen_type_eq} on matrices $U,V,W \in \{-1,0,1\}^{r \times k^2}$
is restrictive: the coefficient matrices must encode a correct matrix product over all rings $R$.
Karstadt and Schwartz introduced a basis-change method that enlarges the set of useful representations
by allowing invertible changes of basis on the inputs and output of \cref{eq:main_vectorized_strassen_type_eq}.
For our purposes, let
\begin{align*}
	B_U, B_V, B_W \in \{-1,0,1\}^{k^2 \times k^2}
\end{align*}
be invertible over the fields under consideration, and define for any valid matrix multiplication
scheme $(U,V,W)$ the basis-changed matrices:
\begin{align}
	U^\prime = U B_U^{-1}, \qquad
	V^\prime = V B_V^{-1}, \qquad
	W^\prime = W B_W^{-\intercal}. \label{eq:transformed_UVW}
\end{align}
Then \cref{eq:main_vectorized_strassen_type_eq} can be rewritten as:
\begin{align}
	\vec{AB} & = W^\intercal \left(\left(U \vec{A}\right) \odot \left(V \vec{B}\right)\right)
	= B_{W} B_{W}^{-1} W^\intercal \left(\left(U B_{U}^{-1} B_{U} \vec{A}\right)
	\odot \left(V B_{V}^{-1} B_{V} \vec{B}\right)\right) \label{eq:apply_basis_change_1}      \\
	         & = B_{W} \left(W^\prime\right)^\intercal
	\left(\left(U^\prime \left( B_{U} \vec{A} \right)\right) \odot
	\left(V^\prime \left( B_{V} \vec{B} \right)\right)\right). \label{eq:apply_basis_change_2}
\end{align}
The transformed matrices $U^\prime$, $V^\prime$, and $W^\prime$ may admit shorter addition
schemes than the original coefficient matrices $U$, $V$, and $W$. Moreover, the preprocessing
maps induced by $B_U$ and $B_V$ and the postprocessing map induced by $B_W$ can be pushed
through the recursion so that their total contribution is $O(n^2\log{n})$ arithmetic operations,
which does not change the leading $O(n^{\log_{k}{r}})$ term when $r > k^2$,
see \Cref{thm:basis_change}.
\begin{theorem} \label{thm:basis_change}
	Let $k \geq 2$ be fixed, and let $n = k^\ell$ for some $\ell \in \mathbb{N}$. Suppose
	$U,V,W \in \{-1, 0, 1\}^{r \times k^{2}}$ form a rank-$r$ multiplication scheme
	satisfying \cref{eq:main_vectorized_strassen_type_eq}.
	Then the basis change induced by any $B_U, B_V, B_W \in \{-1, 0, 1\}^{k^{2} \times k^{2}}$
	that are invertible over the field $\mathbb{F}$ under consideration can be implemented by preprocessing the
	inputs and postprocessing the output, using $O(n^2 \log{n})$ arithmetic operations.
\end{theorem}
\begin{proof}
	For $\ell \geq 1$, let $\mathcal{M}_{\ell}^{U,V,W}(A,B)$ denote the output obtained by recursively applying the multiplication
	scheme determined by $U,V,W$ to matrices $A,B\in\mathbb{F}^{k^\ell\times k^\ell}$.
	Since $U,V,W$ satisfy \cref{eq:main_vectorized_strassen_type_eq}, we have $\mathcal{M}^{U,V,W}_{\ell}(A,B)=AB$.

	Let $\operatorname{BlockVec}_k$ map a $k\times k$ block matrix to the column-wise vector of its $k^2$
	blocks, and let $\operatorname{BlockMat}_k$ denote its inverse. For $Z\in\{U,V,W\}$, define recursively
	a linear operator
	\begin{align*}
		\mathcal{T}^{B_Z}_i:
		\mathbb{F}^{k^i\times k^i}\longrightarrow
		\mathbb{F}^{k^i\times k^i}
	\end{align*}
	by $\mathcal{T}^{B_Z}_{0}(Y) = Y$ and
	\begin{align*}
		\mathcal{T}^{B_Z}_i(Y) = \operatorname{BlockMat}_k\!\left( B_Z\operatorname{BlockVec}_k
		\left(
			\begin{bmatrix}
					\mathcal{T}^{B_Z}_{i-1}(Y_{1,1}) & \cdots &
					\mathcal{T}^{B_Z}_{i-1}(Y_{1,k})                   \\
					\vdots                           & \ddots & \vdots \\
					\mathcal{T}^{B_Z}_{i-1}(Y_{k,1}) & \cdots &
					\mathcal{T}^{B_Z}_{i-1}(Y_{k,k})
				\end{bmatrix}
			\right)
		\right),
	\end{align*}
	where $Y$ is viewed as a $k\times k$ block matrix with blocks in $\mathbb{F}^{k^{i-1}\times k^{i-1}}$.
	Since $B_Z$ is invertible, $\mathcal{T}^{B_Z}_i$ is invertible for every $i$.

	Recall that
	\begin{align*}
		U^\prime = UB_U^{-1},\qquad
		V^\prime = VB_V^{-1},\qquad
		(W^\prime)^\top = B_W^{-1}W^\top.
	\end{align*}
	We claim that, for every $\ell\geq 1$,
	\begin{align}
		\mathcal{T}^{B_W}_\ell \left( \mathcal{M}^{U^\prime,V^\prime,W^\prime}_\ell
		\left( \mathcal{T}^{B_U}_\ell A, \mathcal{T}^{B_V}_\ell B \right) \right)
		= \mathcal{M}^{U,V,W}_\ell(A,B) = AB.
		\label{eq:basis-change-induction}
	\end{align}
	We prove the claim by induction on $\ell$. For $\ell=1$, the input transformation for $A$ satisfies
	\begin{align*}
		U^\prime \operatorname{BlockVec}_k \left(\mathcal{T}^{B_U}_1A\right) =
		UB_U^{-1}B_U\operatorname{BlockVec}_k(A) = U\operatorname{BlockVec}_k(A),
	\end{align*}
	and similarly
	\begin{align*}
		V^\prime \operatorname{BlockVec}_k \left(\mathcal{T}^{B_V}_1B\right) = V\operatorname{BlockVec}_k(B).
	\end{align*}
	Thus the transformed and untransformed schemes form the same vector $M$ of $r$ products. The output
	of the transformed scheme is
	\begin{align*}
		\operatorname{BlockMat}_k\left(B_W^{-1}W^\top M\right).
	\end{align*}
	Applying $\mathcal{T}^{B_W}_1$ gives
	\begin{align*}
		\operatorname{BlockMat}_k \left(B_WB_W^{-1}W^\top M\right) = \operatorname{BlockMat}_k(W^\top M) = AB.
	\end{align*}
	Hence \eqref{eq:basis-change-induction} holds for $\ell=1$.

	Now suppose that the claim holds for $\ell=m-1$. View $A$ and $B$ as $k\times k$ block matrices with
	blocks of size $k^{m-1}\times k^{m-1}$. Let
	\begin{align*}
		(\mathcal{A}_1,\ldots,\mathcal{A}_r)^\top = U\operatorname{BlockVec}_k(A), \qquad
		(\mathcal{B}_1,\ldots,\mathcal{B}_r)^\top = V\operatorname{BlockVec}_k(B)
	\end{align*}
	be the block linear combinations formed by the original scheme. By the definition and linearity of the
	operators $\mathcal{T}^{B_U}_{m-1}$ and $\mathcal{T}^{B_V}_{m-1}$,
	\begin{align*}
		U^\prime \operatorname{BlockVec}_k
		\left(\mathcal{T}^{B_U}_mA\right) = \begin{bmatrix}
			                                    \mathcal{T}^{B_U}_{m-1}(\mathcal{A}_1) \\
			                                    \vdots                                 \\
			                                    \mathcal{T}^{B_U}_{m-1}(\mathcal{A}_r)
		                                    \end{bmatrix},
	\end{align*}
	and
	\begin{align*}
		V^\prime \operatorname{BlockVec}_k \left(\mathcal{T}^{B_V}_mB\right) = \begin{bmatrix}
			                                                                       \mathcal{T}^{B_V}_{m-1}(\mathcal{B}_1) \\
			                                                                       \vdots                                 \\
			                                                                       \mathcal{T}^{B_V}_{m-1}(\mathcal{B}_r)
		                                                                       \end{bmatrix}.
	\end{align*}
	Let $M^{\prime}_{j}$ be the $j$th recursive product computed by the transformed
	scheme. By the induction hypothesis,
	\begin{align*}
		\mathcal{T}^{B_W}_{m-1}(M^{\prime}_{j}) = \mathcal{A}_j\mathcal{B}_j = M_j
	\end{align*}
	for every $j=1,\ldots,r$. The output before postprocessing is therefore
	\begin{align*}
		C^\prime = \operatorname{BlockMat}_k \left(B_W^{-1}W^\top M^\prime\right),
	\end{align*}
	where $M^\prime=(M^\prime_1,\ldots,M^\prime_r)^\top$. Applying $\mathcal{T}^{B_W}_{m}$ and using the linearity of
	$\mathcal{T}^{B_W}_{m-1}$ gives
	\begin{align*}
		\begin{aligned}
			\mathcal{T}^{B_W}_m(C^\prime) & =
			\operatorname{BlockMat}_k \left( B_WB_W^{-1}W^\top
			\begin{bmatrix}
				\mathcal{T}^{B_W}_{m-1}(M^\prime_1) \\
				\vdots                              \\
				\mathcal{T}^{B_W}_{m-1}(M^\prime_r)
			\end{bmatrix}
			\right)                                                                                                    \\
			                              & = \operatorname{BlockMat}_k(W^\top M) = \mathcal{M}^{U,V,W}_{m}(A,B) = AB.
		\end{aligned}
	\end{align*}
	This proves \eqref{eq:basis-change-induction} for every $\ell$.

	It remains to bound the cost of computing the operators $\mathcal{T}^{B_Z}_{\ell}$.
	Let $\beta_Z$ be the number of block additions required to apply $B_Z$ to a vector of $k^2$ blocks.
	Since $k$ is fixed, $\beta_Z = O(1)$. For example, direct evaluation gives $\beta_Z \leq k^4-k^2$.

	At recursion level $i$, there are $k^{2(\ell-i)}$ applications of $B_Z$, each to blocks of size
	$k^{i-1} \times k^{i-1}$. Consequently, the number of scalar additions required is at most
	\begin{align*}
		\beta_Z\sum_{i=1}^{\ell} k^{2(\ell-i)}k^{2(i-1)} = \beta_Z\ell k^{2\ell-2}
		= \frac{\beta_Z}{k^2}n^2\log_k n = O(n^2\log n).
	\end{align*}
	Thus each of the two preprocessing transformations and the postprocessing transformation can be computed
	using $O(n^2\log n)$ arithmetic operations. Since there are only three such transformations, their total
	cost is also $O(n^2\log n)$.
\end{proof}
\Cref{thm:basis_change} tells us that the leading constant in the asymptotic operation
count for a rank-$r$ Strassen-type matrix multiplication scheme depends exclusively
on the number of additions required to compute the linear transformations induced by
$U^\prime$, $V^\prime$, and $W^\prime$.

\subsection{Lifting Schemes} \label{sec:lifting}
Most of the search in our work is performed over $\mathbb{F}_2$, where addition and subtraction
coincide. To obtain an algorithm valid over general fields, the addition schemes found over
$\mathbb{F}_2$ must be lifted to signed addition schemes over a general ring.

We emphasize that this is not a lift of the full matrix multiplication scheme. The
rank-$23$ multiplication scheme is already a valid scheme over any ring $R$. What
must be lifted are the addition schemes used to evaluate the linear maps induced by
$U$, $V$, and $W$ over the field $\mathbb{F}_2$.
To this end, we restrict to lifts that preserve the DAG structure induced by the $\mathbb{F}_2$
addition schemes, and replace each edge weight by a sign in $\{-1, 1\}$, see \Cref{fig:lifting_V_DAG}.
The goal is to decide which additions over $\mathbb{F}_2$ should become subtractions over
the target ring $R$, see \Cref{example:lifting}.
\begin{example} \label{example:lifting}
	Let $x \in R^4$, and let $V \in \{-1,0,1\}^{7 \times 4}$ be given as in
	\Cref{example:UVW_min}. Suppose we want to compute $x \mapsto Vx$ using as few
	additions as possible:
	\begin{align}
		\begin{bmatrix}
			v_1 \\
			v_2 \\
			v_3 \\
			v_4 \\
			v_5 \\
			v_6 \\
			v_7
		\end{bmatrix} = Vx = \begin{bmatrix}
			                     1  & 0 & 0  & 0  \\
			                     0  & 0 & 1  & -1 \\
			                     -1 & 0 & 1  & 0  \\
			                     1  & 0 & -1 & 1  \\
			                     0  & 1 & 0  & 0  \\
			                     0  & 0 & 0  & 1  \\
			                     -1 & 1 & 1  & -1
		                     \end{bmatrix}
		\begin{bmatrix}
			x_1 \\ x_2 \\ x_3 \\ x_4
		\end{bmatrix} = \begin{bmatrix}
			                x_1             \\
			                x_3 - x_4       \\
			                x_3 - x_1       \\
			                x_1 - x_3 + x_4 \\
			                x_2             \\
			                x_4             \\
			                x_2 - x_1 + x_3 - x_4
		                \end{bmatrix}. \label{example:Vx}
	\end{align}
	Over $\mathbb{F}_2$ we find that we can compute the linear map in $4$ additions,
	as follows:
	\begin{align*}
		\begin{alignedat}{4}
			v_{1} & = x_{1}, & \qquad v_{2} & = x_{3} + x_{4}, & \qquad v_{3} & = x_{3} + x_{1}, & \qquad v_{4} & = x_{1} + v_{2}, \\
			v_{5} & = x_{2}, & \qquad v_{6} & = x_{4},         & \qquad v_{7} & = x_{2} + v_{4}. &              &
		\end{alignedat}
	\end{align*}
	We can write this as a DAG with edge-weighted sign variables $\epsilon_i$ that take values in
	$\{-1, 1\}$, see \Cref{fig:lifting_V_DAG}.
	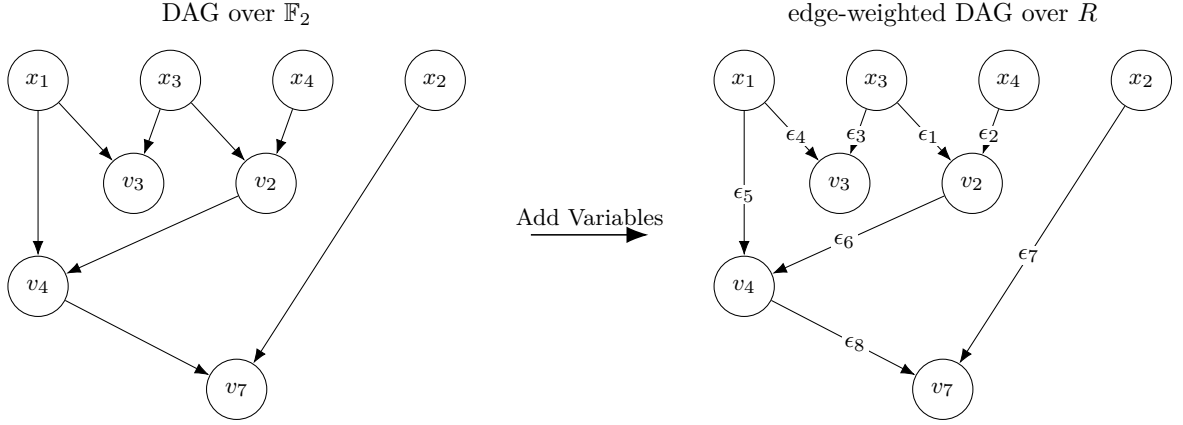
\begin{figure}[H]
		\centering
		\resizebox{\textwidth}{!}{%
			\begin{tikzpicture}[
				vertex/.style={
						shape=circle,
						draw,
						minimum size=8mm,
						text width=8mm,
						align=center,
						inner sep=0pt,
						font=\scriptsize
					},
				edge/.style={-{Latex[length=2mm]}},
				lab/.style={pos=0.58, fill=white, inner sep=1pt, font=\scriptsize}
				]

				\begin{scope}[shift={(0,0)}]
					\node at (0,0.9) {DAG over $\mathbb{F}_2$};

					\node[vertex] (Lx1) at (-2.7,0) {$x_1$};
					\node[vertex] (Lx3) at (-0.9,0) {$x_3$};
					\node[vertex] (Lx4) at ( 0.9,0) {$x_4$};
					\node[vertex] (Lx2) at ( 2.7,0) {$x_2$};

					\node[vertex] (Lv3) at (-1.4,-1.4) {$v_3$};
					\node[vertex] (Lv2) at ( 0.4,-1.4) {$v_2$};

					\node[vertex] (Lv4) at (-2.7,-2.8) {$v_4$};
					\node[vertex] (Lv7) at ( 0.0,-4.2) {$v_7$};

					\draw[edge] (Lx3) -- (Lv2);
					\draw[edge] (Lx4) -- (Lv2);

					\draw[edge] (Lx3) -- (Lv3);
					\draw[edge] (Lx1) -- (Lv3);

					\draw[edge] (Lx1) -- (Lv4);
					\draw[edge] (Lv2) -- (Lv4);

					\draw[edge] (Lx2) -- (Lv7);
					\draw[edge] (Lv4) -- (Lv7);
				\end{scope}

				\draw[-{Latex[length=3mm]}, line width=0.6pt]
				(4.0,-2.1) -- node[above, font=\scriptsize] {Add Variables} (5.6,-2.1);

				\begin{scope}[shift={(9.6,0)}]
					\node at (0,0.9) {edge-weighted DAG over $R$};

					\node[vertex] (Rx1) at (-2.7,0) {$x_1$};
					\node[vertex] (Rx3) at (-0.9,0) {$x_3$};
					\node[vertex] (Rx4) at ( 0.9,0) {$x_4$};
					\node[vertex] (Rx2) at ( 2.7,0) {$x_2$};

					\node[vertex] (Rv3) at (-1.4,-1.4) {$v_3$};
					\node[vertex] (Rv2) at ( 0.4,-1.4) {$v_2$};

					\node[vertex] (Rv4) at (-2.7,-2.8) {$v_4$};
					\node[vertex] (Rv7) at ( 0.0,-4.2) {$v_7$};

					\draw[edge] (Rx3) -- node[lab] {$\epsilon_1$} (Rv2);
					\draw[edge] (Rx4) -- node[lab] {$\epsilon_2$} (Rv2);

					\draw[edge] (Rx3) -- node[lab] {$\epsilon_3$} (Rv3);
					\draw[edge] (Rx1) -- node[lab] {$\epsilon_4$} (Rv3);

					\draw[edge] (Rx1) -- node[lab] {$\epsilon_5$} (Rv4);
					\draw[edge] (Rv2) -- node[lab] {$\epsilon_6$} (Rv4);

					\draw[edge] (Rx2) -- node[lab] {$\epsilon_7$} (Rv7);
					\draw[edge] (Rv4) -- node[lab] {$\epsilon_8$} (Rv7);
				\end{scope}

			\end{tikzpicture}%
		}
		\caption{Lifting an unweighted DAG over $\mathbb{F}_2$ to an edge-weighted DAG over an arbitrary ring $R$.
			The left DAG computes the linear transformation induced by $V$ over $\mathbb{F}_2$.
			On the right, each edge is assigned a sign variable $\epsilon_i \in \{-1,1\}$.
			Choosing the signs appropriately can lift the computation from $\mathbb{F}_2$ to $R$.
			Direct copies such as $v_1 = x_1$, $v_5 = x_2$, and $v_6 = x_4$ are omitted in this graph.}
		\label{fig:lifting_V_DAG}
	\end{figure}
	Using the variable edge-weighted DAG in \Cref{fig:lifting_V_DAG}, and comparing it with the
	direct expansion of $Vx$ shown in \cref{example:Vx}, we get:
	\begin{align}
		\begin{bmatrix}
			x_1                                                                  \\
			\epsilon_1 x_3 + \epsilon_2 x_4                                      \\
			\epsilon_4 x_1 + \epsilon_3 x_3                                      \\
			\epsilon_5 x_1 + \epsilon_6\epsilon_1 x_3 + \epsilon_6\epsilon_2 x_4 \\
			x_2                                                                  \\
			x_4                                                                  \\
			\epsilon_8\epsilon_5 x_1 + \epsilon_7 x_2
			+ \epsilon_8\epsilon_6\epsilon_1 x_3
			+ \epsilon_8\epsilon_6\epsilon_2 x_4
		\end{bmatrix} = \begin{bmatrix}
			                v_1 \\
			                v_2 \\
			                v_3 \\
			                v_4 \\
			                v_5 \\
			                v_6 \\
			                v_7
		                \end{bmatrix} = \begin{bmatrix}
			                                x_1             \\
			                                x_3 - x_4       \\
			                                x_3 - x_1       \\
			                                x_1 - x_3 + x_4 \\
			                                x_2             \\
			                                x_4             \\
			                                x_2 - x_1 + x_3 - x_4
		                                \end{bmatrix},
	\end{align}
	Comparing coefficients yields the following polynomial equations in the variables
	$\epsilon_i \in \{-1, 1\}$:
	\begin{align}
		\begin{alignedat}{4}
			\epsilon_1 & = 1,                                  & \qquad \epsilon_2           & = -1, & \qquad \epsilon_3 & = 1, & \qquad \epsilon_4 & = -1, \\
			\epsilon_5 & = 1,                                  & \qquad \epsilon_6\epsilon_1 & = -1,
			           & \qquad \epsilon_6\epsilon_2           & = 1,                        &       &                                                      \\
			\epsilon_7 & = 1,                                  & \qquad \epsilon_8\epsilon_5 & = -1,
			           & \qquad \epsilon_8\epsilon_6\epsilon_1 & = 1,
			           & \qquad \epsilon_8\epsilon_6\epsilon_2 & = -1.
		\end{alignedat}
	\end{align}
	A valid solution is
	\begin{align}
		\epsilon_1 = 1,\quad
		\epsilon_2 = -1,\quad
		\epsilon_3 = 1,\quad
		\epsilon_4 = -1,\quad
		\epsilon_5 = 1,\quad
		\epsilon_6 = -1,\quad
		\epsilon_7 = 1,\quad
		\epsilon_8 = -1.
	\end{align}
	This solution coincides with the DAG illustrated for computing the linear
	map $\vec{B} \mapsto V\vec{B}$ in \Cref{fig:DAG_rep_minadd_example}.
\end{example}
In general, we solve the polynomial equations arising from our lifting setup using a
simple backtracking approach that systematically tests all possible values of each
variable $\epsilon_i$, pruning the search whenever one of the resulting reduced equations
gives rise to a contradiction.

\section{Search Method}
Previous work on reducing the number of additions required to compute existing rank-$23$ schemes for multiplying
$3 \times 3$ matrices has considered only relatively few well-known schemes \cite{MartenssonWagner_NumberBeast}.
We make use of the database of over $17,000$ rank-$23$ multiplication schemes due to Heule et al. \cite{HeuleKauersSeidl2021NewWays3x3}.

Each of the schemes found by Heule et al. was first computed in $\mathbb{F}_2$, and then most of them
were lifted to general fields $\mathbb{F}$ by systematically changing some of the $1$s to $-1$s \cite{HeuleKauersSeidl2021NewWays3x3}.
We can therefore easily revert each $\mathbb{F}$ scheme to the $\mathbb{F}_2$ scheme from which it was
derived by replacing each entry of $U$, $V$, and $W$ by its absolute value, that is, changing each $-1$ to a $1$.

\begin{example}
	The first rank-$23$ scheme for $3 \times 3$ matrix multiplication was introduced by Laderman
	in 1976 \cite{Laderman1976Noncommutative3x3}. In the form of \cref{eq:main_vectorized_strassen_type_eq}, Laderman's scheme
	can be written in terms of the matrices illustrated in \Cref{fig:laderman_matrices}.
	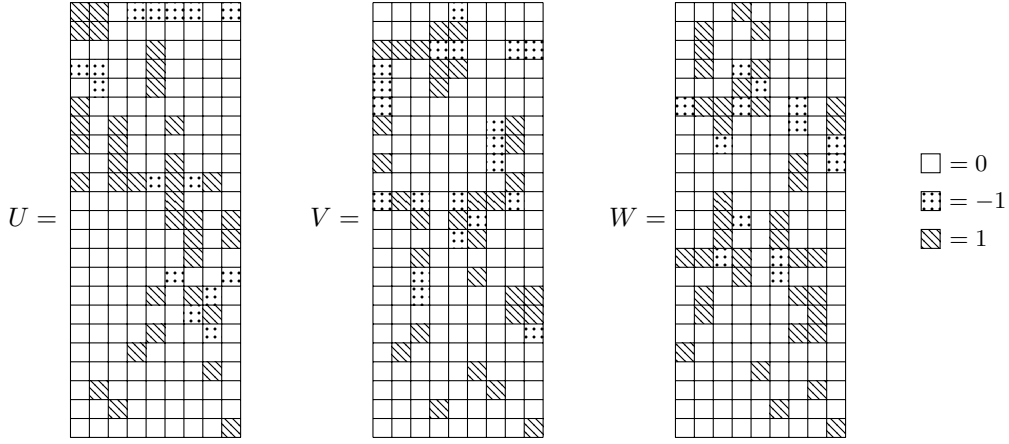
\begin{figure}[H]
		\begin{center}
			\begin{tikzpicture}[scale=0.25]
				\begin{scope}
					\node at (-2,11.65) {$U = $};
					\draw (0,0) grid (9, 23);
					\fill[decorate, pattern=north west lines] (0,22) rectangle (1,23);
					\fill[decorate, pattern=north west lines] (1,22) rectangle (2,23);
					\fill[decorate, pattern=dots] (3,22) rectangle (4,23);
					\fill[decorate, pattern=dots] (4,22) rectangle (5,23);
					\fill[decorate, pattern=dots] (5,22) rectangle (6,23);
					\fill[decorate, pattern=dots] (6,22) rectangle (7,23);
					\fill[decorate, pattern=dots] (8,22) rectangle (9,23);
					\fill[decorate, pattern=north west lines] (0,21) rectangle (1,22);
					\fill[decorate, pattern=north west lines] (1,21) rectangle (2,22);
					\fill[decorate, pattern=north west lines] (4,20) rectangle (5,21);
					\fill[decorate, pattern=dots] (0,19) rectangle (1,20);
					\fill[decorate, pattern=dots] (1,19) rectangle (2,20);
					\fill[decorate, pattern=north west lines] (4,19) rectangle (5,20);
					\fill[decorate, pattern=dots] (1,18) rectangle (2,19);
					\fill[decorate, pattern=north west lines] (4,18) rectangle (5,19);
					\fill[decorate, pattern=north west lines] (0,17) rectangle (1,18);
					\fill[decorate, pattern=north west lines] (0,16) rectangle (1,17);
					\fill[decorate, pattern=north west lines] (2,16) rectangle (3,17);
					\fill[decorate, pattern=north west lines] (5,16) rectangle (6,17);
					\fill[decorate, pattern=north west lines] (0,15) rectangle (1,16);
					\fill[decorate, pattern=north west lines] (2,15) rectangle (3,16);
					\fill[decorate, pattern=north west lines] (2,14) rectangle (3,15);
					\fill[decorate, pattern=north west lines] (5,14) rectangle (6,15);
					\fill[decorate, pattern=north west lines] (0,13) rectangle (1,14);
					\fill[decorate, pattern=north west lines] (2,13) rectangle (3,14);
					\fill[decorate, pattern=north west lines] (3,13) rectangle (4,14);
					\fill[decorate, pattern=dots] (4,13) rectangle (5,14);
					\fill[decorate, pattern=north west lines] (5,13) rectangle (6,14);
					\fill[decorate, pattern=dots] (6,13) rectangle (7,14);
					\fill[decorate, pattern=north west lines] (7,13) rectangle (8,14);
					\fill[decorate, pattern=north west lines] (5,12) rectangle (6,13);
					\fill[decorate, pattern=north west lines] (5,11) rectangle (6,12);
					\fill[decorate, pattern=north west lines] (6,11) rectangle (7,12);
					\fill[decorate, pattern=north west lines] (8,11) rectangle (9,12);
					\fill[decorate, pattern=north west lines] (6,10) rectangle (7,11);
					\fill[decorate, pattern=north west lines] (8,10) rectangle (9,11);
					\fill[decorate, pattern=north west lines] (6,9) rectangle (7,10);
					\fill[decorate, pattern=dots] (5,8) rectangle (6,9);
					\fill[decorate, pattern=dots] (8,8) rectangle (9,9);
					\fill[decorate, pattern=north west lines] (4,7) rectangle (5,8);
					\fill[decorate, pattern=north west lines] (6,7) rectangle (7,8);
					\fill[decorate, pattern=dots] (7,7) rectangle (8,8);
					\fill[decorate, pattern=dots] (6,6) rectangle (7,7);
					\fill[decorate, pattern=north west lines] (7,6) rectangle (8,7);
					\fill[decorate, pattern=north west lines] (4,5) rectangle (5,6);
					\fill[decorate, pattern=dots] (7,5) rectangle (8,6);
					\fill[decorate, pattern=north west lines] (3,4) rectangle (4,5);
					\fill[decorate, pattern=north west lines] (7,3) rectangle (8,4);
					\fill[decorate, pattern=north west lines] (1,2) rectangle (2,3);
					\fill[decorate, pattern=north west lines] (2,1) rectangle (3,2);
					\fill[decorate, pattern=north west lines] (8,0) rectangle (9,1);
				\end{scope}

				\begin{scope}[xshift=16cm]
					\node at (-2,11.65) {$V = $};
					\draw (0,0) grid (9, 23);
					\fill[decorate, pattern=dots] (4,22) rectangle (5,23);
					\fill[decorate, pattern=north west lines] (3,21) rectangle (4,22);
					\fill[decorate, pattern=north west lines] (4,21) rectangle (5,22);
					\fill[decorate, pattern=north west lines] (0,20) rectangle (1,21);
					\fill[decorate, pattern=north west lines] (1,20) rectangle (2,21);
					\fill[decorate, pattern=north west lines] (2,20) rectangle (3,21);
					\fill[decorate, pattern=dots] (3,20) rectangle (4,21);
					\fill[decorate, pattern=dots] (4,20) rectangle (5,21);
					\fill[decorate, pattern=dots] (7,20) rectangle (8,21);
					\fill[decorate, pattern=dots] (8,20) rectangle (9,21);
					\fill[decorate, pattern=dots] (0,19) rectangle (1,20);
					\fill[decorate, pattern=north west lines] (3,19) rectangle (4,20);
					\fill[decorate, pattern=north west lines] (4,19) rectangle (5,20);
					\fill[decorate, pattern=dots] (0,18) rectangle (1,19);
					\fill[decorate, pattern=north west lines] (3,18) rectangle (4,19);
					\fill[decorate, pattern=dots] (0,17) rectangle (1,18);
					\fill[decorate, pattern=north west lines] (0,16) rectangle (1,17);
					\fill[decorate, pattern=dots] (6,16) rectangle (7,17);
					\fill[decorate, pattern=north west lines] (7,16) rectangle (8,17);
					\fill[decorate, pattern=dots] (6,15) rectangle (7,16);
					\fill[decorate, pattern=north west lines] (7,15) rectangle (8,16);
					\fill[decorate, pattern=north west lines] (0,14) rectangle (1,15);
					\fill[decorate, pattern=dots] (6,14) rectangle (7,15);
					\fill[decorate, pattern=north west lines] (7,13) rectangle (8,14);
					\fill[decorate, pattern=dots] (0,12) rectangle (1,13);
					\fill[decorate, pattern=north west lines] (1,12) rectangle (2,13);
					\fill[decorate, pattern=dots] (2,12) rectangle (3,13);
					\fill[decorate, pattern=dots] (4,12) rectangle (5,13);
					\fill[decorate, pattern=north west lines] (5,12) rectangle (6,13);
					\fill[decorate, pattern=north west lines] (6,12) rectangle (7,13);
					\fill[decorate, pattern=dots] (7,12) rectangle (8,13);
					\fill[decorate, pattern=north west lines] (2,11) rectangle (3,12);
					\fill[decorate, pattern=north west lines] (4,11) rectangle (5,12);
					\fill[decorate, pattern=dots] (5,11) rectangle (6,12);
					\fill[decorate, pattern=dots] (4,10) rectangle (5,11);
					\fill[decorate, pattern=north west lines] (5,10) rectangle (6,11);
					\fill[decorate, pattern=north west lines] (2,9) rectangle (3,10);
					\fill[decorate, pattern=dots] (2,8) rectangle (3,9);
					\fill[decorate, pattern=north west lines] (5,8) rectangle (6,9);
					\fill[decorate, pattern=dots] (2,7) rectangle (3,8);
					\fill[decorate, pattern=north west lines] (7,7) rectangle (8,8);
					\fill[decorate, pattern=north west lines] (8,7) rectangle (9,8);
					\fill[decorate, pattern=north west lines] (7,6) rectangle (8,7);
					\fill[decorate, pattern=north west lines] (8,6) rectangle (9,7);
					\fill[decorate, pattern=north west lines] (2,5) rectangle (3,6);
					\fill[decorate, pattern=dots] (8,5) rectangle (9,6);
					\fill[decorate, pattern=north west lines] (1,4) rectangle (2,5);
					\fill[decorate, pattern=north west lines] (5,3) rectangle (6,4);
					\fill[decorate, pattern=north west lines] (6,2) rectangle (7,3);
					\fill[decorate, pattern=north west lines] (3,1) rectangle (4,2);
					\fill[decorate, pattern=north west lines] (8,0) rectangle (9,1);
				\end{scope}

				\begin{scope}[xshift=32cm]
					\node at (-2,11.65) {$W = $};
					\draw (0,0) grid (9, 23);
					\fill[decorate, pattern=north west lines] (3,22) rectangle (4,23);
					\fill[decorate, pattern=north west lines] (1,21) rectangle (2,22);
					\fill[decorate, pattern=north west lines] (4,21) rectangle (5,22);
					\fill[decorate, pattern=north west lines] (1,20) rectangle (2,21);
					\fill[decorate, pattern=north west lines] (1,19) rectangle (2,20);
					\fill[decorate, pattern=dots] (3,19) rectangle (4,20);
					\fill[decorate, pattern=north west lines] (4,19) rectangle (5,20);
					\fill[decorate, pattern=north west lines] (3,18) rectangle (4,19);
					\fill[decorate, pattern=dots] (4,18) rectangle (5,19);
					\fill[decorate, pattern=dots] (0,17) rectangle (1,18);
					\fill[decorate, pattern=north west lines] (1,17) rectangle (2,18);
					\fill[decorate, pattern=north west lines] (2,17) rectangle (3,18);
					\fill[decorate, pattern=dots] (3,17) rectangle (4,18);
					\fill[decorate, pattern=north west lines] (4,17) rectangle (5,18);
					\fill[decorate, pattern=dots] (6,17) rectangle (7,18);
					\fill[decorate, pattern=north west lines] (8,17) rectangle (9,18);
					\fill[decorate, pattern=north west lines] (2,16) rectangle (3,17);
					\fill[decorate, pattern=dots] (6,16) rectangle (7,17);
					\fill[decorate, pattern=north west lines] (8,16) rectangle (9,17);
					\fill[decorate, pattern=dots] (2,15) rectangle (3,16);
					\fill[decorate, pattern=dots] (8,15) rectangle (9,16);
					\fill[decorate, pattern=north west lines] (6,14) rectangle (7,15);
					\fill[decorate, pattern=dots] (8,14) rectangle (9,15);
					\fill[decorate, pattern=north west lines] (6,13) rectangle (7,14);
					\fill[decorate, pattern=north west lines] (2,12) rectangle (3,13);
					\fill[decorate, pattern=north west lines] (2,11) rectangle (3,12);
					\fill[decorate, pattern=dots] (3,11) rectangle (4,12);
					\fill[decorate, pattern=north west lines] (5,11) rectangle (6,12);
					\fill[decorate, pattern=north west lines] (2,10) rectangle (3,11);
					\fill[decorate, pattern=north west lines] (5,10) rectangle (6,11);
					\fill[decorate, pattern=north west lines] (0,9) rectangle (1,10);
					\fill[decorate, pattern=north west lines] (1,9) rectangle (2,10);
					\fill[decorate, pattern=dots] (2,9) rectangle (3,10);
					\fill[decorate, pattern=north west lines] (3,9) rectangle (4,10);
					\fill[decorate, pattern=dots] (5,9) rectangle (6,10);
					\fill[decorate, pattern=north west lines] (6,9) rectangle (7,10);
					\fill[decorate, pattern=north west lines] (7,9) rectangle (8,10);
					\fill[decorate, pattern=north west lines] (3,8) rectangle (4,9);
					\fill[decorate, pattern=dots] (5,8) rectangle (6,9);
					\fill[decorate, pattern=north west lines] (1,7) rectangle (2,8);
					\fill[decorate, pattern=north west lines] (6,7) rectangle (7,8);
					\fill[decorate, pattern=north west lines] (7,7) rectangle (8,8);
					\fill[decorate, pattern=north west lines] (1,6) rectangle (2,7);
					\fill[decorate, pattern=north west lines] (7,6) rectangle (8,7);
					\fill[decorate, pattern=north west lines] (6,5) rectangle (7,6);
					\fill[decorate, pattern=north west lines] (7,5) rectangle (8,6);
					\fill[decorate, pattern=north west lines] (0,4) rectangle (1,5);
					\fill[decorate, pattern=north west lines] (4,3) rectangle (5,4);
					\fill[decorate, pattern=north west lines] (7,2) rectangle (8,3);
					\fill[decorate, pattern=north west lines] (5,1) rectangle (6,2);
					\fill[decorate, pattern=north west lines] (8,0) rectangle (9,1);
				\end{scope}

				\begin{scope}[xshift=45cm]
					\draw[decorate, pattern=north west lines] (0,10) rectangle (1,11);
					\node at (2.5, 10.5) {\small $= 1$};

					\draw[decorate, pattern=dots] (0,12) rectangle (1,13);
					\node at (3, 12.5) {\small $= -1$};

					\draw[decorate] (0,14) rectangle (1,15);
					\node at (2.5, 14.5) {\small $= 0$};
				\end{scope}
			\end{tikzpicture}
		\end{center}
		\caption{Illustration of the matrices $U$, $V$, and $W$ that constitute Laderman's rank-$23$ scheme for
			multiplying $3 \times 3$ matrices \protect\cite{Laderman1976Noncommutative3x3}. These specific matrices
			$U$, $V$, and $W$ were adapted from the Laderman scheme provided by \protect\cite{HeuleKauersSeidl2021NewWays3x3}.} \label{fig:laderman_matrices}
	\end{figure}
	The matrices $U, V, W \in \{-1, 0, 1\}^{23 \times 9}$ constituting Laderman's scheme, shown in
	\Cref{fig:laderman_matrices}, define a valid $3 \times 3$ matrix multiplication scheme over general fields
	$\mathbb{F}$. If we replace each $-1$ with a $1$, the resulting scheme remains valid over $\mathbb{F}_2$. The converse does not
	necessarily hold: a scheme over $\mathbb{F}_2$ does not necessarily lift to a ternary scheme over general
	fields \cite{HeuleKauersSeidl2021NewWays3x3}.
\end{example}

\subsection{Greedy-Potential}
The greedy-potential heuristic designed by Mårtensson and Wagner in their work \cite{MartenssonWagner_NumberBeast}
is a method for reducing the number of additions required to evaluate a linear map
\begin{align}
	x \longmapsto Gx,
	\qquad
	G \in \{-1,0,1\}^{m\times n}.
\end{align}
In our application, this map is one of the three linear maps induced by
$U$, $V$, and $W^\top$, or by their basis-changed analogues. The
heuristic constructs an addition scheme by repeatedly introducing a new
temporary variable that is the signed sum of two variables already available.

Suppose that the current set of available variables is
$x_{1}, \ldots, x_{n}$, and consider a candidate
\begin{align*}
	x_{n+1} = x_{i} + \sigma x_{j},
	\qquad 1 \leq i < j \leq n,\quad \sigma \in \{-1,1\}.
\end{align*}
This candidate is useful in every row $k$ for which the terms involving
$x_i$ and $x_j$ occur with the same relative sign as in $x_{n+1}$.
That is, the candidate is useful in each row of $G$ indexed by $k \in I_{i,j,\sigma}$, where
\begin{align*}
	I_{i,j,\sigma}(G) = \{\, k \in [m] : G_{k,i} = \sigma G_{k,j} \in \{-1,1\} \,\}.
\end{align*}
For every $k \in I_{i,j,\sigma}(G)$, the two terms $G_{k,i}x_i + G_{k,j}x_j$ can be replaced by the single term
$G_{k,i}x_{n+1}$. Thus, if $G^\prime \in \{-1,0,1\}^{m\times(n+1)}$ denotes the augmented matrix after
introducing $x_{n+1}$, see \Cref{fig:heuristic_algo}, then
\begin{align*}
	G^\prime_{k,n+1} = G_{k,i}, \qquad
	G^\prime_{k,i} = G^\prime_{k,j} = 0,
	\qquad k\in I_{i,j,\sigma}(G),
\end{align*}
while all other entries are copied from $G$. Introducing the variable $x_{n+1}$
costs one addition, but removes one addition from each row in $I_{i,j,\sigma}(G)$.
Hence the immediate saving from adding this candidate is
\begin{align*}
	|I_{i,j,\sigma}(G)| - 1.
\end{align*}
\Cref{fig:heuristic_algo} illustrates how $G$ is augmented according to the
greedy-potential algorithm, when choosing a pair of indices $1 \leq i,j \leq n$
for constructing the variable $x_{n+1} = x_{i} + \sigma x_{j}$.
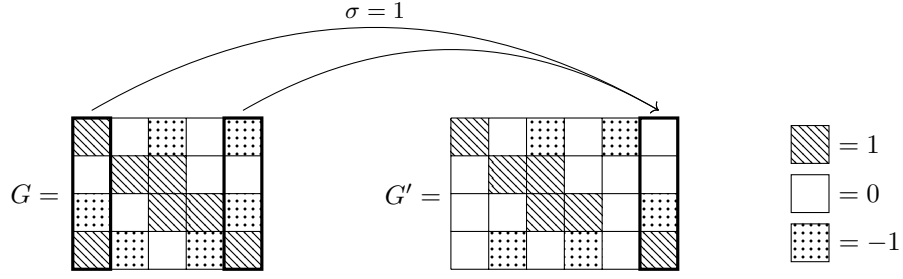
\begin{figure}[H]
	\begin{center}
		\begin{tikzpicture}[scale=1.0]
			\node at (-0.5,1) {$G = $};
			\draw[step=0.5] (0,0) grid (2.5,2);
			\fill[pattern=north west lines, pattern color=black] (0,1.5) rectangle (0.5,2);
			\fill[pattern=dots, pattern color=black] (0,0.5) rectangle (0.5,1);
			\fill[pattern=north west lines, pattern color=black] (0,0) rectangle (0.5,0.5);
			\fill[pattern=north west lines, pattern color=black] (0.5,1) rectangle (1,1.5);
			\fill[pattern=dots, pattern color=black] (0.5,0) rectangle (1,0.5);
			\fill[pattern=dots, pattern color=black] (1,1.5) rectangle (1.5,2);
			\fill[pattern=north west lines, pattern color=black] (1,1) rectangle (1.5,1.5);
			\fill[pattern=north west lines, pattern color=black] (1,0.5) rectangle (1.5,1);
			\fill[pattern=north west lines, pattern color=black] (1.5,0.5) rectangle (2,1);
			\fill[pattern=dots, pattern color=black] (1.5,0) rectangle (2,0.5);
			\fill[pattern=dots, pattern color=black] (2,1.5) rectangle (2.5,2);
			\fill[pattern=dots, pattern color=black] (2,0.5) rectangle (2.5,1);
			\fill[pattern=north west lines, pattern color=black] (2,0) rectangle (2.5,0.5);
			\draw[very thick] (0,0) rectangle (0.5,2);
			\draw[very thick] (2,0) rectangle (2.5,2);

			\node at (4.5,1) {$G^\prime = $};
			\draw[step=0.5] (4.999,0) grid (8,2);
			\fill[pattern=north west lines, pattern color=black] (5,1.5) rectangle (5.5,2);
			\fill[pattern=north west lines, pattern color=black] (5.5,1) rectangle (6,1.5);
			\fill[pattern=dots, pattern color=black] (5.5,0) rectangle (6,0.5);
			\fill[pattern=dots, pattern color=black] (6,1.5) rectangle (6.5,2);
			\fill[pattern=north west lines, pattern color=black] (6,1) rectangle (6.5,1.5);
			\fill[pattern=north west lines, pattern color=black] (6,0.5) rectangle (6.5,1);
			\fill[pattern=north west lines, pattern color=black] (6.5,0.5) rectangle (7,1);
			\fill[pattern=dots, pattern color=black] (6.5,0) rectangle (7,0.5);
			\fill[pattern=dots, pattern color=black] (7,1.5) rectangle (7.5,2);
			\fill[pattern=dots, pattern color=black] (7.5, 0.5) rectangle (8, 1);
			\fill[pattern=north west lines, pattern color=black] (7.5, 0) rectangle (8, 0.5);
			\draw[very thick] (7.5,0) rectangle (8,2);

			\draw[->] (0.25,2.1) to[bend left] node[midway, above] {\small $\sigma = 1$} (7.75,2.1);
			\draw[->] (2.25,2.1) to[bend left] (7.75,2.1);

			\begin{scope}[xshift=5.5cm]
				\fill[pattern=north west lines, pattern color=black, draw=black] (4,1.4) rectangle (4.5,1.9);
				\node at (4.9, 1.65) {$= 1$};
				\fill[color=white, draw=black] (4,0.75) rectangle (4.5,1.25);
				\node at (4.9, 1) {$= 0$};
				\fill[pattern=dots, pattern color=black, draw=black] (4,0.1) rectangle (4.5,0.6);
				\node at (5.04, 0.35) {$= -1$};
			\end{scope}
		\end{tikzpicture}
	\end{center}
	\caption{Illustration of how the column vector $G_{:,n+1}^\prime$ is constructed, and how the
		matrix $G$ is augmented, when $\sigma = 1$ and $(i, j) = (1, 5)$.}\label{fig:heuristic_algo}
\end{figure}
The potential of the matrix $G \in \{-1, 0, 1\}^{m \times n}$ is defined as
\begin{align}
	P(G) = \sum_{1 \leq i < j \leq n} \sum_{\sigma \in \{-1,1\}}
	\max\{0,\, |I_{i,j,\sigma}(G)| - 1\}. \label{eq:potential_def}
\end{align}
In other words, this quantity is the sum of the positive immediate savings over all candidate variables
$x_{n+1} = x_{i} + \sigma x_{j}$. It provides an upper bound on the total number of additions that
can be saved by repeatedly taking steps of the greedy-potential algorithm. For a fixed parameter
$\alpha > 0$, we score the candidate $x_{n+1} = x_{i} + \sigma x_{j}$ by
\begin{align*}
	\operatorname{score}(i,j,\sigma) = |I_{i,j,\sigma}(G)| - 1 + \alpha P(G^\prime).
\end{align*}
The greedy-potential heuristic considers only candidates with positive immediate savings,
i.e. those satisfying $|I_{i,j,\sigma}(G)| \geq 2$. Among these candidates, it applies one
with the largest $\operatorname{score}(i,j,\sigma)$. The process is repeated until no such
candidate remains, equivalently until the potential is zero.
The remaining rows of the final matrix are then expanded directly.
Together, the chosen temporary variables and these final
row expansions give an addition scheme for the original map $x \mapsto Gx$.

\subsubsection{A Potential Difference}
Currently, to the best of our knowledge, all implementations of the greedy-potential heuristic
implement finding the potential of a matrix more or less directly from the definition, see
\cref{eq:potential_def}.
Observe that if our candidate variable is $x_{n+1} = x_{i} + \sigma x_{j}$, then for all pairs
of column indices $i^\prime, j^\prime \in [n] \setminus \{i, j\}$ we have
\begin{align*}
	I_{i^\prime,j^\prime,\sigma}(G^\prime) = I_{i^\prime,j^\prime,\sigma}(G),
\end{align*}
since, among the existing columns, augmenting $G$ according to the candidate variable
$x_{n+1} = x_{i} + \sigma x_{j}$ affects only columns $i$ and $j$.
Therefore, computing $P(G^\prime)$ reduces to simply looking at those index pairs $(i^\prime, j^\prime)$
that have $\{i, j\} \cap \{i^\prime, j^\prime\} \neq \emptyset$. Define $N_{i,j}(G)$ as the
sum of the nonnegative immediate savings associated with the candidates
$x_{n+1} = x_{i} + x_{j}$ and $x_{n+1} = x_{i} - x_{j}$,
i.e. we define
\begin{align*}
	N_{i,j}(G) = \sum_{\sigma \in \{-1, 1\}} \max\left(0, |I_{i,j,\sigma}(G)| - 1\right),
\end{align*}
so that:
\begin{align*}
	P(G^\prime)  = \sum_{i^\prime = 1}^{n} \sum_{j^\prime = i^\prime+1}^{n + 1} N_{i^\prime, j^\prime}(G^\prime)
	= P(G) & + \left(\sum_{\substack{k = 1                             \\ k \neq i}}^{n} N_{i,k}(G^\prime) +
	\sum_{\substack{k = 1                                              \\ k \neq j}}^{n} N_{k,j}(G^\prime)
	- N_{i,j}(G^\prime) + \sum_{k = 1}^{n} N_{k,n+1}(G^\prime) \right) \\
	       & - \left(\sum_{\substack{k = 1                             \\ k \neq j}}^{n} N_{k,j}(G)
	+ \sum_{\substack{k = 1                                            \\ k \neq i}}^{n} N_{k,i}(G)  - N_{i,j}(G)\right).
\end{align*}
If we know the value of $P(G)$, computing the right-hand side requires only $O(mn)$ operations, whereas
computing \cref{eq:potential_def} directly requires $O(mn^2)$ operations. We confirm the effectiveness
of a greedy-potential implementation based on this new setup in \Cref{fig:gp_comparison}.
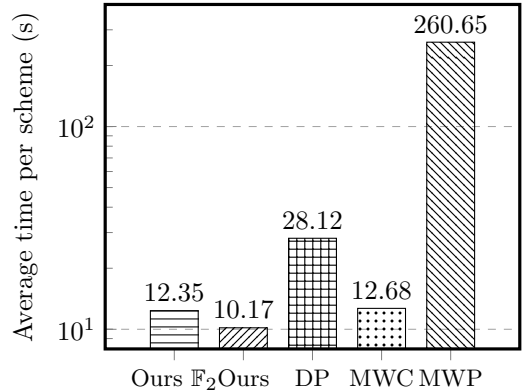
\begin{figure}[H]
	\begin{center}
		\begin{tikzpicture}
			\begin{axis}[
					ybar,
					ymode=log,
					log origin=infty,
					ymin=8,
					ymax=400,
					bar width=18pt,
					enlarge x limits=0.25,
					ylabel={Average time per scheme (s)},
					symbolic x coords={oursftwo,ours,dp,mwc,mwp},
					xtick={oursftwo,ours,dp,mwc,mwp},
					xticklabels={\small{Ours $\mathbb{F}_2$}, \small{Ours}, \small{DP}, \small{MWC}, \small{MWP}},
					axis lines=box,
					xtick pos=left,
					ytick pos=left,
					axis line style={black, very thick},
					ymajorgrids=true,
					grid style={black!35, dashed},
					nodes near coords,
					nodes near coords align={vertical},
					bar shift=0pt,
					point meta=rawy,
					nodes near coords={
							\pgfmathprintnumber[fixed,precision=2]{\pgfplotspointmeta}
						},
					every node near coord/.append style={text=black},
					scale=0.8
				]

				\addplot+[draw=black, fill=white, pattern=horizontal lines]
				coordinates {(oursftwo, 12.35)};

				\addplot+[draw=black, fill=white, pattern=north east lines]
				coordinates {(ours, 10.17)};

				\addplot+[draw=black, fill=white, pattern=grid]
				coordinates {(dp, 28.115)};

				\addplot+[draw=black, fill=white, pattern=dots]
				coordinates {(mwc, 12.681)};

				\addplot+[draw=black, fill=white, pattern=north west lines]
				coordinates {(mwp, 260.645)};

			\end{axis}
		\end{tikzpicture}

	\end{center}
	\caption{Speed comparison between implementations of the greedy potential heuristic.
		\texttt{DP} is the implementation by \protect\cite{perminov2025parallel},
		\texttt{MWC} and \texttt{MWP} are the C and Python implementations, respectively, provided
		by \protect\cite{MartenssonWagner_NumberBeast}.
		The comparison was conducted on the $3 \times 3$ matrix multiplication scheme dataset provided by
		\protect\cite{HeuleKauersSeidl2021NewWays3x3}.
		Note that our $\mathbb{F}_2$ implementation operates over $\mathbb{F}_2$, while all
		other implementations work over $\{-1, 0, 1\}$.
		Each run was allotted $72$ hours on a single core of an Intel Xeon E5-2660 v3 CPU
		to process the entire dataset. Note that the \texttt{MWP} and \texttt{DP} methods failed
		to process the entire dataset within the time limit.
	}\label{fig:gp_comparison}
\end{figure}
We remark that, over $\mathbb{F}_2$, we can optimize the running time of the dense greedy-potential method by
using bitwise operations. The matrices in \cite{HeuleKauersSeidl2021NewWays3x3} are relatively sparse, so
certain $\mathbb{F}_2$ optimizations might not be relevant for our purposes, see \Cref{fig:gp_comparison}.

\section{Results}
The focus of this work is on rank-$23$ general $3 \times 3$ matrix multiplication, so $r = 23$ and
$k = 3$. For each $\rho = 2^{-6},2^{-5},2^{-4},2^{-3},2^{-2},2^{-1}$ we draw $2^{14} = 16384$ invertible
Bernoulli-random $k^2 \times k^2$ matrices. We augment each data point in the dataset provided by Heule et al.
with these matrices, as detailed in \Cref{thm:basis_change}.
Finally, for each greedy-potential parameter $\alpha = 0.1,0.2,\dots,1.0$ we apply the greedy-potential
method to each possibly basis-changed data point.

\begin{figure}[H]
	\centering
	\begin{subfigure}[t]{0.49\textwidth}
		\centering
		\includegraphics[width=\linewidth]{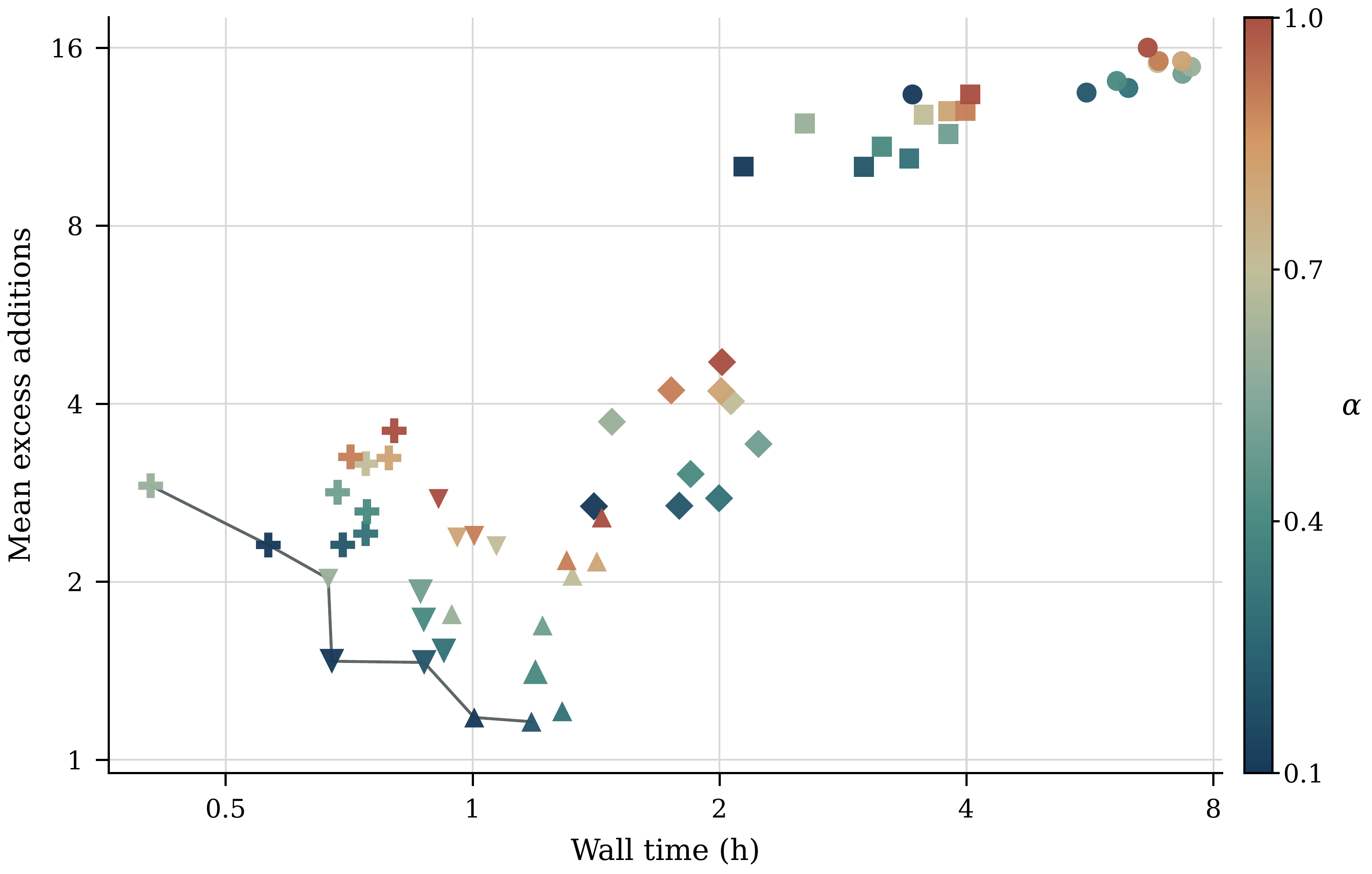}
		\caption{Full sweep.}
	\end{subfigure}
	\hfill
	\begin{subfigure}[t]{0.49\textwidth}
		\centering
		\includegraphics[width=\linewidth]{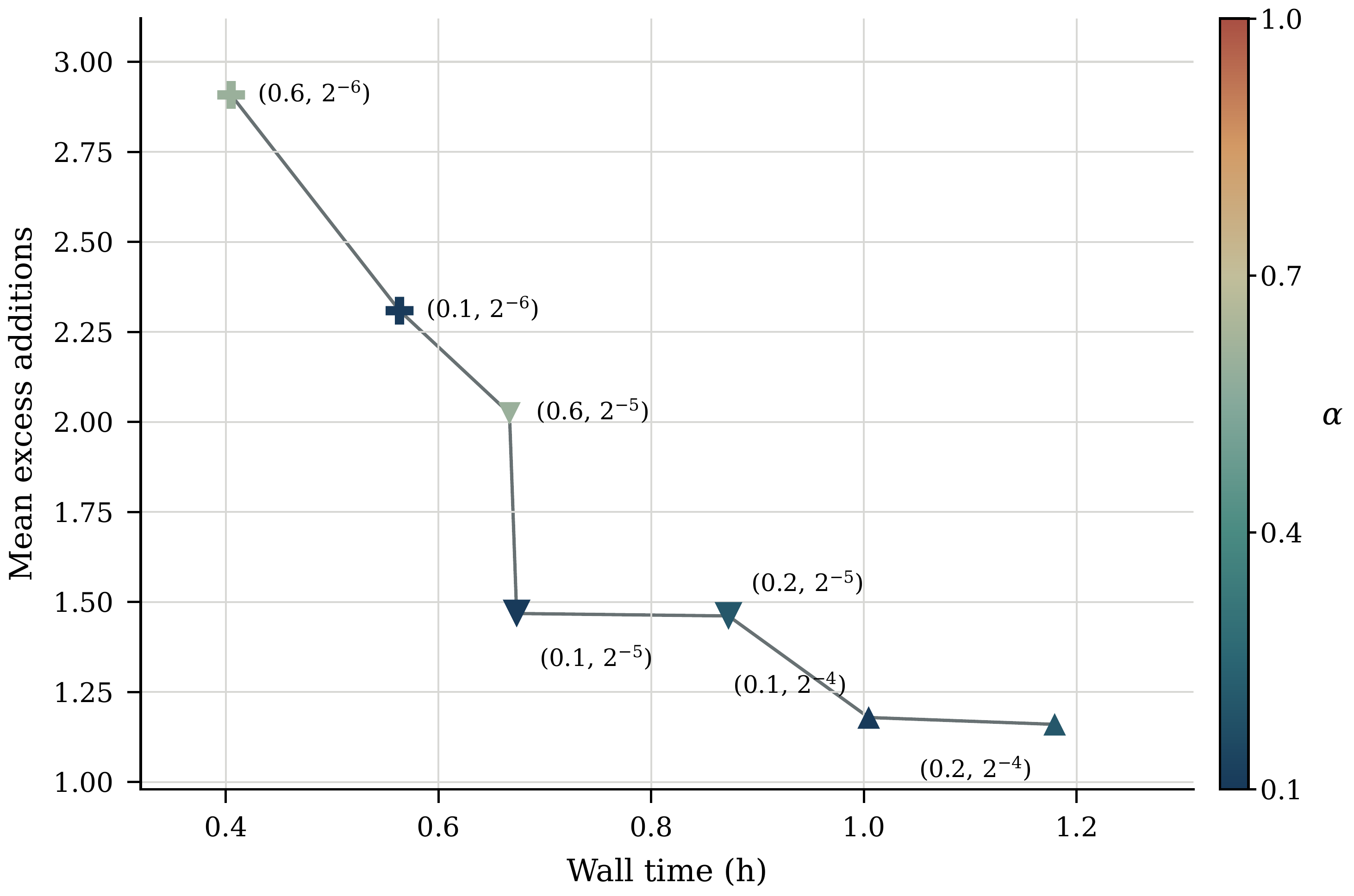}
		\caption{Pareto frontier.}
		\label{fig:quality_time_alpha_rho_tradeoff_B}
	\end{subfigure}
	\par\smallskip
	\sharedRhoLegend
	\par\smallskip
	\caption{Quality--time trade-off for the $(\alpha,\rho)$ sweep using eight cores.
		Color encodes $\alpha$, marker shape encodes $\rho$, and the line shows the Pareto frontier.
		Excess additions are measured relative to the best result found for a given data point in the dataset provided
		by \cite{HeuleKauersSeidl2021NewWays3x3}.}
	\label{fig:quality_time_alpha_rho_tradeoff}
\end{figure}
The experiments in \Cref{fig:quality_time_alpha_rho_tradeoff} collectively give rise
to six distinct $55$-addition schemes over $\mathbb{F}_2$, each of which improves upon
the previous state of the art of $56$ additions established by \cite{sun2026exact56additionrank23}.

Based on the results in \Cref{fig:quality_time_alpha_rho_tradeoff}, we selected $32$
cores, $4,000,000$ basis-change matrices, $\rho = 2^{-4}$, and $\alpha = 0.2$. In other words, we
allocated more computational resources to the rightmost point on the Pareto-optimal lower envelope shown
in \Cref{fig:quality_time_alpha_rho_tradeoff_B}. This produced schemes with addition counts
as shown in \Cref{fig:3M_GP_run}.
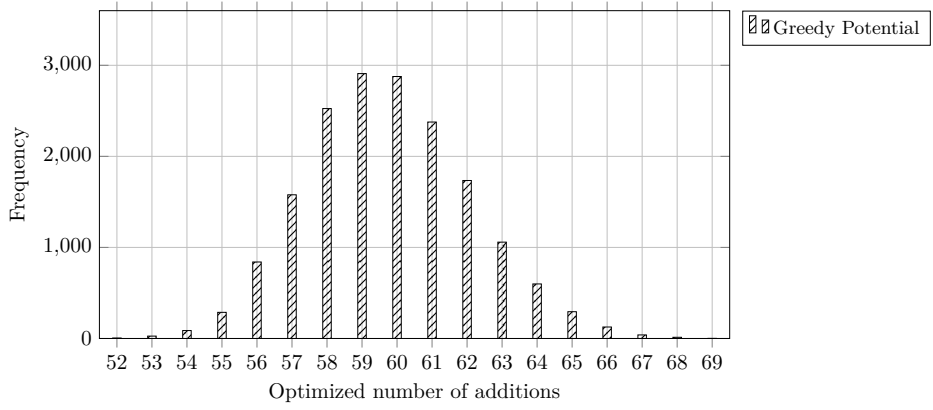
\begin{figure}[H]
	\begin{center}
		\begin{tikzpicture}[scale=0.8]
			\begin{axis}[
					width=12cm,
					height=7cm,
					ybar,
					bar width=4pt,
					xlabel={Optimized number of additions},
					ylabel={Frequency},
					xmin=51.5, xmax=69.5,
					ymin=0,
					ymax=3600,
					grid=major,
					xtick={52,53,...,71},
					legend cell align={left},
					legend style={
							at={(1.02,1)},
							anchor=north west,
							draw=black,
							fill=white,
							font=\small
						},
				]

				\addplot[
					draw=black,
					pattern=north east lines,
					pattern color=black,
				] coordinates {
						(52,5)
						(53,26)
						(54,87)
						(55,287)
						(56,840)
						(57,1578)
						(58,2525)
						(59,2909)
						(60,2877)
						(61,2377)
						(62,1735)
						(63,1058)
						(64,599)
						(65,294)
						(66,126)
						(67,39)
						(68,13)
						(69,1)
					};
				\addlegendentry{Greedy Potential}
			\end{axis}
		\end{tikzpicture}
	\end{center}
	\caption{Distribution of optimized addition counts over $\mathbb{F}_2$ for the rank-$23$,
	$3 \times 3$ matrix-multiplication schemes in the dataset of
	\protect\cite{HeuleKauersSeidl2021NewWays3x3}. For each scheme, the count is the total number
	of additions required to evaluate $U\vec{A}$, $V\vec{B}$, and $W^\intercal M$ in
	\protect\cref{eq:main_vectorized_strassen_type_eq}, using the best basis changes found among
	$4,000,000$ invertible matrices in $\mathbb{F}_2^{3^2 \times 3^2}$ obtained from
	Bernoulli-random matrices with bias $\rho = 2^{-4}$.
	Addition schemes were constructed using the greedy-potential heuristic with $\alpha = 0.2$,
	the search used $32$ cores and ran for $\sim 50$ hours.} \label{fig:3M_GP_run}
\end{figure}
Of the schemes found, $405$ are state of the art in that they require $55$ or fewer additions to compute,
but they only work over $\mathbb{F}_2$.
We extend one of the resulting $52$-addition schemes to work over general fields $\mathbb{F}$ using the method described
in \Cref{sec:lifting}, which leads to the following scheme, derived from the datapoint
\texttt{17x2+52x+4y6+5y4+14y3+z32+2z27/i12w195c23ci-000.tab}, in \cite{HeuleKauersSeidl2021NewWays3x3}.

The inverses of the basis-change matrices $B_U$, $B_V$, and $B_W$ are
\begin{align*}
	 & B_{U}^{-1} = \begin{bmatrix}
		                0 & -1 & 0 & 0  & -1 & 0  & 0 & 0 & 0  \\
		                0 & 1  & 0 & 0  & 0  & 0  & 1 & 0 & 0  \\
		                0 & 0  & 0 & 0  & 1  & 0  & 0 & 0 & 0  \\
		                0 & 0  & 0 & 0  & 0  & 0  & 0 & 0 & -1 \\
		                0 & 0  & 0 & -1 & 0  & 0  & 0 & 0 & 0  \\
		                0 & 0  & 0 & 1  & 0  & 0  & 0 & 1 & 1  \\
		                1 & 0  & 0 & 0  & 0  & -1 & 0 & 0 & 0  \\
		                0 & 0  & 1 & 0  & 0  & 0  & 0 & 0 & 0  \\
		                0 & 0  & 0 & 0  & 0  & 1  & 0 & 0 & 0
	                \end{bmatrix}
	 & B_{V}^{-1} = \begin{bmatrix}
		                0 & 0 & 0 & 0 & 0 & 1 & 0 & 0 & 0 \\
		                0 & 0 & 1 & 0 & 0 & 0 & 0 & 0 & 0 \\
		                0 & 1 & 0 & 0 & 0 & 1 & 0 & 0 & 0 \\
		                1 & 0 & 0 & 0 & 0 & 0 & 0 & 0 & 0 \\
		                0 & 0 & 0 & 0 & 1 & 0 & 0 & 0 & 1 \\
		                0 & 0 & 0 & 1 & 0 & 0 & 0 & 0 & 0 \\
		                0 & 0 & 0 & 0 & 0 & 0 & 1 & 0 & 0 \\
		                0 & 0 & 0 & 0 & 0 & 0 & 0 & 0 & 1 \\
		                0 & 0 & 0 & 1 & 0 & 0 & 0 & 1 & 0
	                \end{bmatrix}
	 &                                           & B_{W}^{-1} =
	\begin{bmatrix}
		0 & 0 & 0  & 0 & 0  & 0  & 0  & 1 & 0 \\
		0 & 0 & -1 & 0 & 0  & 0  & 0  & 0 & 1 \\
		0 & 1 & 0  & 0 & 0  & 0  & 0  & 0 & 0 \\
		0 & 0 & -1 & 0 & 0  & 0  & 0  & 0 & 0 \\
		0 & 0 & 0  & 1 & 0  & 0  & 0  & 0 & 0 \\
		0 & 1 & 0  & 0 & 1  & -1 & 0  & 0 & 0 \\
		0 & 0 & 0  & 0 & -1 & 0  & 0  & 0 & 0 \\
		1 & 0 & 0  & 0 & 0  & 0  & -1 & 0 & 0 \\
		0 & 0 & 0  & 0 & 0  & 0  & -1 & 0 & 0
	\end{bmatrix}
\end{align*}
We first compute the basis change induced by $B_{U}$ and $B_{V}$ to obtain
$\tilde{A}, \tilde{B} \in \mathbb{F}^{3^\ell \times 3^\ell}$ from $A,B \in \mathbb{F}^{3^\ell \times 3^\ell}$
for some $\ell \in \mathbb{N}$, see the proof of \Cref{thm:basis_change}.
We compute the basis-changed $U^\prime$, $V^\prime$, and $W^\prime$, see \cref{eq:transformed_UVW}.
We compute $U^\prime \vec{\tilde{A}}$ in $12$ additions as follows:
\begin{align*}
	 & u_{1} = \tilde{A}_{1,3} + \tilde{A}_{3,1}, &  & u_{2} = \tilde{A}_{2,3} - \tilde{A}_{3,3},  &  & u_{3} = \tilde{A}_{1,1} - u_{2}, \\
	 & u_{4} = \tilde{A}_{1,2} + \tilde{A}_{2,2}, &  & u_{5} = \tilde{A}_{3,3},                    &  & u_{6} = u_{3},                   \\
	 & u_{7} = u_{4},                             &  & u_{8} = u_{5},                              &  & u_{9} = \tilde{A}_{1,1} + u_{1}, \\
	 & u_{10} = \tilde{A}_{1,2} + u_{1},          &  & u_{11} = \tilde{A}_{3,2} - u_{2},           &  & u_{12} = \tilde{A}_{1,3},        \\
	 & u_{13} = \tilde{A}_{3,1} - u_{17},         &  & u_{14} = \tilde{A}_{2,1},                   &  & u_{15} = u_{11} - u_{12},        \\
	 & u_{16} = \tilde{A}_{3,1},                  &  & u_{17} = \tilde{A}_{2,1} - \tilde{A}_{1,2}, &  & u_{18} = \tilde{A}_{2,2},        \\
	 & u_{19} = u_{14} + u_{11},                  &  & u_{20} = \tilde{A}_{2,3},                   &  & u_{21} = \tilde{A}_{3,2},        \\
	 & u_{22} = \tilde{A}_{1,2},                  &  & u_{23} = u_{21} + u_{17}.
\end{align*}
Similarly, we find the following $12$-addition scheme for computing $V^\prime \vec{\tilde{B}}$:
\begin{align*}
	 & v_{1} = \tilde{B}_{2,1} - \tilde{B}_{1,2}, &  & v_{2} = \tilde{B}_{3,2} + v_{1},            &  & v_{3} = \tilde{B}_{2,3},           \\
	 & v_{4} = \tilde{B}_{3,3} - \tilde{B}_{1,3}, &  & v_{5} = v_{4},                              &  & v_{6} = \tilde{B}_{1,3},           \\
	 & v_{7} = v_{3} - v_{2},                     &  & v_{8} = v_{7},                              &  & v_{9} = \tilde{B}_{1,2},           \\
	 & v_{10} = \tilde{B}_{1,1} - v_{1},          &  & v_{11} = \tilde{B}_{1,1} + \tilde{B}_{2,2}, &  & v_{12} = \tilde{B}_{1,1},          \\
	 & v_{13} = \tilde{B}_{2,2} + v_{17},         &  & v_{14} = \tilde{B}_{2,2},                   &  & v_{15} = \tilde{B}_{3,3} + v_{11}, \\
	 & v_{16} = \tilde{B}_{2,1},                  &  & v_{17} = \tilde{B}_{3,3} - \tilde{B}_{2,1}, &  & v_{18} = \tilde{B}_{3,1} - v_{2},  \\
	 & v_{19} = \tilde{B}_{3,1} - v_{11},         &  & v_{20} = v_{18},                            &  & v_{21} = \tilde{B}_{3,3},          \\
	 & v_{22} = \tilde{B}_{3,2},                  &  & v_{23} = \tilde{B}_{3,2} - v_{17}.
\end{align*}
Finally, let $M_{i} = u_{i} v_{i}$. We compute $(W^\prime)^\intercal M$
using $28$ additions:
\begin{align*}
	 & s_{1} = M_{7} - M_{16},   &  & s_{2} = M_{1} - M_{12},   &  & s_{3} = M_{14} - M_{5},   \\
	 & s_{4} = M_{8} - M_{2},    &  & s_{5} = M_{17} - M_{8},   &  & s_{6} = M_{13} + M_{4},   \\
	 & s_{7} = M_{22} + s_{1},   &  & s_{8} = M_{16} - s_{2},   &  & s_{9} = s_{3} - M_{21},   \\
	 & s_{10} = s_{4} - M_{20},  &  & s_{11} = M_{21} + s_{5},  &  & s_{12} = M_{14} + s_{6},  \\
	 & s_{13} = M_{10} - s_{8},  &  & s_{14} = M_{9} - s_{8},   &  & s_{15} = s_{10} - M_{11}, \\
	 & s_{16} = s_{10} - M_{12}, &  & s_{17} = M_{22} + s_{11}, &  & s_{18} = M_{16} + s_{12}, \\
	 & s_{19} = M_{18} + s_{13}, &  & s_{20} = s_{14} - M_{20}, &  & s_{21} = M_{14} - s_{15}, \\
	 & s_{22} = s_{16} - M_{3},  &  & s_{23} = M_{6} + s_{16},  &  & s_{24} = M_{23} + s_{17}, \\
	 & s_{25} = M_{17} + s_{18}, &  & s_{26} = M_{19} + s_{21}, &  & s_{27} = s_{23} - M_{15}, \\
	 & s_{28} = M_{19} - s_{27}, &  & w_{1} = s_{24},           &  & w_{2} = s_{20},           \\
	 & w_{3} = s_{26},           &  & w_{4} = s_{22},           &  & w_{5} = s_{25},           \\
	 & w_{6} = s_{28},           &  & w_{7} = s_{9},            &  & w_{8} = s_{19},           \\
	 & w_{9} = s_{7}.
\end{align*}
Define the linear operator $T_{\ell}^{B_W}$ as in the proof of \Cref{thm:basis_change}. We get the
intermediate matrix:
\begin{align*}
	\tilde{C} = T_{\ell}^{B_W^{-1}} AB = \begin{bmatrix}
		                                     w_1 & w_4 & w_7 \\
		                                     w_2 & w_5 & w_8 \\
		                                     w_3 & w_6 & w_9
	                                     \end{bmatrix}.
\end{align*}
Postprocessing $\tilde{C}$ using the recursive linear transformation induced by $B_W$
yields exactly the matrix product $C = AB$, see \Cref{app:verification} for a symbolic verification.

\section{Conclusion}
In this work, we accelerate the greedy-potential heuristic by computing the change in potential associated
with each candidate addition, rather than recomputing the full potential from its definition.
Using this approach, we found hundreds of distinct rank-$23$ Strassen-type schemes for general $3 \times 3$
matrix multiplication over $\mathbb{F}_2$ whose addition counts improve upon the previous state of the art.
In particular, our search produced schemes requiring as few as $52$ additions.

We then lifted one such $52$-addition scheme from $\mathbb{F}_2$ to a signed addition scheme valid over
arbitrary fields. Together with the associated change of basis, this yields a rank-$23$ recursive
matrix-multiplication scheme requiring $52$ additions or subtractions at each recursive step.

\backmatter

\bmhead{Acknowledgments}
I am grateful to Jakob Lemvig and Christian Henriksen for their supervision,
feedback on methodology, and helpful comments.

\clearpage
\appendix
\section{Reproducibility}
All code, build configuration, and job scripts used to produce the results in
this paper are available in our GitHub repository:
\url{https://github.com/dragonoverlord3000/shortest_linear_program}.

The benchmark data were generated by running the following commands from the repository root:
\begin{lstlisting}[style=bashstyle]
make download-3x3

make clean
OMP_NUM_THREADS=32 make
OMP_NUM_THREADS=32 make bench-full BENCH_ARGS="--seed=6283 --verbose --no-preprocess --no-postprocess --potential_alpha=0.2 --optimization_strategy=single_shot --search_method=greedy_potential --benchmarks=3x3_matmul --num_basis_change=4000000 --potential_bit_p=0.0625 --output=build/benchmarks/full_3x3_c32_gp_6283.json"
\end{lstlisting}

\section{Verification}
\label{app:verification}
The following Python program provides a symbolic verification of the rank-$23$
scheme. It constructs the basis-changed inputs, evaluates the $23$ products and
the output straight-line program, reverses the output basis change, and verifies
with SymPy that the resulting matrix is exactly $AB$.

\begin{lstlisting}[
	style=pythonstyle,
	caption={Symbolic verification script (\texttt{small\_verify.py}).},
	label={lst:small-verify}
]
"""
We remark for the reader, that Heule et. al. use a transposed version of the matrices
for their dataset, so even though my SLP algorithms represent elements in column major
format, parts of this algorithm is in row-major format.
"""

import sympy as sp

BU = sp.Matrix(
    [
        [0, -1, 0, 0, -1, 0, 0, 0, 0],
        [0,  1, 0, 0,  0, 0, 1, 0, 0],
        [0,  0, 0, 0,  1, 0, 0, 0, 0],
        [0,  0, 0, 0,  0, 0, 0, 0, -1],
        [0,  0, 0, -1, 0, 0, 0, 0, 0],
        [0,  0, 0, 1,  0, 0, 0, 1, 1],
        [1,  0, 0, 0,  0, -1, 0, 0, 0],
        [0,  0, 1, 0,  0, 0, 0, 0, 0],
        [0,  0, 0, 0,  0, 1, 0, 0, 0],
    ]
)

BV = sp.Matrix(
    [
        [0, 0, 0, 0, 0, 1, 0, 0, 0],
        [0, 0, 1, 0, 0, 0, 0, 0, 0],
        [0, 1, 0, 0, 0, 1, 0, 0, 0],
        [1, 0, 0, 0, 0, 0, 0, 0, 0],
        [0, 0, 0, 0, 1, 0, 0, 0, 1],
        [0, 0, 0, 1, 0, 0, 0, 0, 0],
        [0, 0, 0, 0, 0, 0, 1, 0, 0],
        [0, 0, 0, 0, 0, 0, 0, 0, 1],
        [0, 0, 0, 1, 0, 0, 0, 1, 0],
    ]
)

BW = sp.Matrix(
    [
        [0, 0,  0, 0,  0,  0,  0, 1, 0],
        [0, 0, -1, 0,  0,  0,  0, 0, 1],
        [0, 1,  0, 0,  0,  0,  0, 0, 0],
        [0, 0, -1, 0,  0,  0,  0, 0, 0],
        [0, 0,  0, 1,  0,  0,  0, 0, 0],
        [0, 1,  0, 0,  1, -1,  0, 0, 0],
        [0, 0,  0, 0, -1,  0,  0, 0, 0],
        [1, 0,  0, 0,  0,  0, -1, 0, 0],
        [0, 0,  0, 0,  0,  0, -1, 0, 0],
    ]
)

A_vec = sp.Matrix([sp.Symbol(f"A_{i, j}", commutative=False) for i in range(3) for j in range(3)])
B_vec = sp.Matrix([sp.Symbol(f"B_{i, j}", commutative=False) for i in range(3) for j in range(3)])

A_vec_prime = BU.inv() @ A_vec
B_vec_prime = BV.inv() @ B_vec


def compute_U_prime(A_vec):
    A11, A12, A13, A21, A22, A23, A31, A32, A33 = list(A_vec)

    u1 = A13 + A31
    u2 = A23 - A33
    u3 = A11 - u2
    u4 = A12 + A22
    u9 = A11 + u1
    u10 = A12 + u1
    u11 = A32 - u2
    u17 = A21 - A12
    u13 = A31 - u17
    u15 = u11 - A13
    u19 = A21 + u11
    u23 = A32 + u17

    return {
        1: u1, 2: u2, 3: u3, 4: u4, 5: A33, 6: u3, 7: u4, 8: A33, 9: u9, 10: u10, 11: u11, 12: A13, 13: u13, 
        14: A21, 15: u15, 16: A31, 17: u17, 18: A22, 19: u19, 20: A23, 21: A32, 22: A12, 23: u23
    }


def compute_V_prime(B_vec):
    B11, B12, B13, B21, B22, B23, B31, B32, B33 = list(B_vec)

    v1 = B21 - B12
    v2 = B32 + v1
    v4 = B33 - B13
    v7 = B23 - v2
    v10 = B11 - v1
    v11 = B11 + B22
    v17 = B33 - B21
    v13 = B22 + v17
    v15 = B33 + v11
    v18 = B31 - v2
    v19 = B31 - v11
    v23 = B32 - v17

    return { 
        1: v1, 2: v2, 3: B23, 4: v4, 5: v4, 6: B13, 7: v7, 8: v7, 9: B12, 10: v10, 11: v11, 12: B11,
        13: v13, 14: B22, 15: v15, 16: B21, 17: v17, 18: v18, 19: v19, 20: v18, 21: B33, 22: B32, 23: v23
    }

def compute_W_prime(M):
    s1 = M[7] - M[16]
    s2 = M[1] - M[12]
    s3 = M[14] - M[5]
    s4 = M[8] - M[2]
    s5 = M[17] - M[8]
    s6 = M[13] + M[4]
    s7 = M[22] + s1
    s8 = M[16] - s2
    s9 = s3 - M[21]
    s10 = s4 - M[20]
    s11 = M[21] + s5
    s12 = M[14] + s6
    s13 = M[10] - s8
    s14 = M[9] - s8
    s15 = s10 - M[11]
    s16 = s10 - M[12]
    s17 = M[22] + s11
    s18 = M[16] + s12
    s19 = M[18] + s13
    s20 = s14 - M[20]
    s21 = M[14] - s15
    s22 = s16 - M[3]
    s23 = M[6] + s16
    s24 = M[23] + s17
    s25 = M[17] + s18
    s26 = M[19] + s21
    s27 = s23 - M[15]
    s28 = M[19] - s27

    return {1: s24, 2: s20, 3: s26, 4: s22, 5: s25, 6: s28, 7: s9, 8: s19, 9: s7}


u = compute_U_prime(A_vec_prime)
v = compute_V_prime(B_vec_prime)

M = {i: u[i] * v[i] for i in range(1, 24)}

w_prime = compute_W_prime(M)
w_prime_vec = sp.Matrix([w_prime[i] for i in range(1, 10)])
C_vec_col_major = BW.inv() @ w_prime_vec

C = sp.expand(
    sp.Matrix(
        3,
        3,
        lambda i, j: C_vec_col_major[3 * j + i],
    )
)

print(C)

# easier to verify this way
A = sp.Matrix([[sp.Symbol(f"A_{i, j}", commutative=False) for j in range(3)] for i in range(3)])
B = sp.Matrix([[sp.Symbol(f"B_{i, j}", commutative=False) for j in range(3)] for i in range(3)])
diff = sp.simplify(sp.expand(C - A @ B))
print("-" * 32)
print(diff)
\end{lstlisting}

\bibliography{main}
\end{document}